\documentclass[a4paper,12pt]{article}
\usepackage[left=2.8 cm,top=2.0cm,bottom=3.4cm,right=2.8cm]{geometry}
\usepackage{amsfonts,amssymb,amsthm,amsmath}
\usepackage{mathrsfs}
\usepackage{bbm}
\usepackage{fontspec}

\usepackage{xcolor}

\definecolor{green(munsell)}{rgb}{0.0, 0.66, 0.47}
\definecolor{BlueGreenn}{rgb}{0.3,0.5,0.8}
\definecolor{DB}{rgb}{0.3,0.3,0.3}
\definecolor{DOr}{rgb}{0.7,0.3,0.3}
\definecolor{DGr}{rgb}{0.3,0.7,0.3}
\definecolor{DBl}{rgb}{0.1,0.3,0.5}
\definecolor{arylideyellow}{rgb}{0.91, 0.84, 0.42}
\definecolor{burntorange}{rgb}{0.8, 0.33, 0.0}
\definecolor{chromeyellow}{rgb}{1.0, 0.65, 0.0}
\definecolor{niceyellow}{RGB}{255,204,0}
\definecolor{lightgray}{rgb}{0.96,0.96,0.96}
\definecolor{darkgreen}{rgb}{0.1,0.45,0.06}
\definecolor{darkred}{rgb}{0.65,0.1,0.1}
\definecolor{lightred}{rgb}{1,0.9,0.9}
\definecolor{darkblue}{rgb}{0.1,0.1,0.5}
\definecolor{lightblue}{rgb}{0.88,0.92,0.99}
\definecolor{lightorange}{rgb}{1,0.95,0.89}
\definecolor{lightyellow}{rgb}{1,1,0.87}
\definecolor{burntlightorange}{RGB}{245,235,211}
\definecolor{skde}{RGB}{0,80,158}
\definecolor{skdelight}{RGB}{198,215,231}
\definecolor{skdeverylight}{RGB}{228,245,251}
\definecolor{skdelys}{RGB}{122,178,220}
\definecolor{pink}{RGB}{232,88,208}

\usepackage[hypertexnames=true, citecolor = burntorange, colorlinks = true, linkcolor = BlueGreenn, pdffitwindow = false, urlbordercolor = white,pdfborder={0 0 0}]{hyperref}%

\usepackage[T1]{fontenc}
\usepackage{microtype}

\usepackage{color}
\usepackage[utf8]{inputenc}

\usepackage{framed}
\usepackage{centernot}

\usepackage{tikz-cd}
\usepackage{tkz-euclide}

\usepackage[shortlabels]{enumitem}
\usepackage{dirtytalk}
\usepackage{subcaption}
\usepackage{float}
\usepackage{physics}
\usepackage{tcolorbox}
\usepackage{tikz}
\usetikzlibrary{
  math,
  bbox,
  knots,
  hobby,
  decorations.pathreplacing,
  shapes.geometric,
  calc
}

\usepackage{comment}

\usepackage{pgfplots}
\pgfplotsset{compat=1.18} 

\usepackage{mathtools}
\mathtoolsset{showonlyrefs=true} 

\usepackage{ytableau}
\ytableausetup{boxsize=0.5em,centertableaux}

\usepackage[normalem]{ulem}

\setlist{nolistsep}
\usepackage{titlesec}
\titleformat{\section}{\large\bfseries\color{DBl}}{\thesection}{1em}{}
\titleformat{\subsection}[runin]
{\normalfont\normalsize\bfseries\color{DBl}}{\thesubsection}{0.5em}{}[.]

\titlespacing\subsection{0pt}{12pt plus 4pt minus 2pt}{4pt plus 2pt minus 2pt}

\makeatletter
\patchcmd{\@maketitle}
  {\@title}
  {\color{\@titlecolor}\@title}
  {}{}
\newcommand{\@titlecolor}{DBl}
\newcommand{\titlecolor}[1]{\renewcommand{\@titlecolor}{#1}}
\makeatother

\numberwithin{equation}{section}

\newtheorem{theorem}{Theorem}[section]
\newtheorem{proposition}[theorem]{Proposition}
\newtheorem{lemma}[theorem]{Lemma}

\newtheorem{corollary}[theorem]{Corollary}

\newtheorem{example}[theorem]{Example}

\theoremstyle{definition}
\newtheorem{definition}[theorem]{Definition}

\usepackage{bm}       
\usepackage{fontspec}

\usepackage{euscript} 

\def\cc{\operatorname{c}}
\newcommand{\Std}{\operatorname{Std}}

\newcommand{\col}{\,:\,}
\newcommand{\perpr}{{\perp_{\Rl}}}

\usepackage{amsmath}

\DeclareMathOperator*{\slim}{s-lim}

\DeclareMathOperator*{\newspan}{span}

\DeclareMathOperator{\imag}{Im}
\let\real\relax
\DeclareMathOperator{\real}{Re}

\DeclareMathOperator{\Aut}{Aut}

\DeclareMathOperator{\supp}{supp}

\newcommand{\Cl}{\mathbb{C}}
\newcommand{\Rl}{\mathbb{R}}
\newcommand{\Nl}{\mathbb{N}}

\newcommand{\B}{\mathcal{B}}

\newcommand{\Hil}{\mathcal{H}}

\newcommand{\CA}[0]{\mathcal{A}} \newcommand{\CB}[0]{\mathcal{B}}

\newcommand{\CE}[0]{\mathcal{E}} \newcommand{\CF}[0]{\mathcal{F}}

\newcommand{\CK}[0]{\mathcal{K}} \newcommand{\CL}[0]{\mathcal{L}}
\newcommand{\CM}[0]{\mathcal{M}} \newcommand{\CN}[0]{\mathcal{N}}
\newcommand{\CO}[0]{\mathcal{O}} \newcommand{\CP}[0]{\mathcal{P}}
\newcommand{\CQ}[0]{\mathcal{Q}}

\newcommand{\CW}[0]{\mathcal{W}}

\newcommand{\CC}{\mathcal{C}}

\DeclareMathOperator{\Prob}{PM}

\newcommand{\bB}{{\mathbb B}}

\newcommand{\bp}{{\bm{p}}}

\newcommand{\bx}{{\bm{x}}}
\newcommand{\by}{{\bm{y}}}
\newcommand{\bz}{{\bm{z}}}

\newcommand{\Om}{\Omega}
\newcommand{\om}{\omega}
\newcommand{\la}{\lambda}
\newcommand{\eps}{\varepsilon}

\newcommand{\La}{\Lambda}

\newcommand{\SO}{{\mathrm{SO}}}
\newcommand{\AU}{{\mathrm{AU}}}

\newcommand{\sa}{{\mathrm sa}}

\DeclareMathAlphabet\EuScript{U}{eus}{m}{n}

\newcommand{\Bor}{\EuScript{B}}
\newcommand{\Meas}{\EuScript{M}}
\newcommand{\Eu}{\mathbf{E}}
\newcommand{\Poi}{\mathbf{P}}

\setlist[itemize,1]{label=$\bullet$}

\usepackage[style=alphabetic, labelnumber, defernumbers=true, firstinits=true, maxbibnames=99, url=false, doi=true, isbn=false, eprint=false, backend=biber]{biblatex}
\renewbibmacro{in:}{}
\DeclareFieldFormat{pages}{#1}
\newcommand{\leqs}{\leq_{\mathrm s}}

\newcommand{\vees}{\vee_{\mathrm s}}
\newcommand{\wedges}{\wedge_{\mathrm s}}
\DeclareMathOperator{\NW}{NW}

\title{\textbf{Causal quantum-mechanical localization observables in lattices of real projections}}

\author{Gandalf Lechner\footnote{Friedrich-Alexander-Universität Erlangen-Nürnberg, Department Mathematik, Cauerstr. 11, 91058 Erlangen, Germany, {\ttfamily gandalf.lechner@fau.de}} \and Ivan Romualdo de Oliveira\footnote{Departamento de Física, Universidade Federal de Lavras, Caixa Postal 3037, 37203-202, Lavras, MG, Brazil, {\ttfamily ivanromualdo@usp.br}}}

\date{\small February 11, 2026}

\begin{document}

\maketitle

\begin{abstract}
    Quantum-mechanical observables for spatial and spacetime localization are considered from a lattice-theoretic perspective. It is shown that when replacing the lattice of all complex orthogonal projections underlying the Born rule by the lattice of real linear projections with symplectic complementation, the well-known No-Go theorems of Hegerfeldt and Malament no longer apply: Causal and Poincaré covariant localization observables exist. In this setting, several features of quantum field theory, such as Lorentz symmetry and modular localization, emerge automatically. In the case of a particle described by a massive positive energy representation of the Poincaré group, the Brunetti-Guido-Longo map defines a spacetime localization observable that is unique under some natural further assumptions. Regarding possible probabilistic interpretations of such a structure, a Gleason theorem and a cluster theorem for symplectic complements are established. These imply that evaluating such localization observables in states yields a fuzzy probability measure that fails to be a measure because it is not additive. However, for separation scales that are large in comparison to the Compton wavelength, the emerging modular localization picture is essentially additive and approximates the one of Newton-Wigner.
\end{abstract}

\section{Introduction}\label{Section:Introduction}

The Born rule of quantum mechanics associates with any experimentally verifyable proposition about a quantum system (of the form ``measurement of observable $X$ returns a value in the set $A\subset \Rl$'') the probability of this proposition being true, depending on the state of the system. It can naturally be formulated in the language of lattices and logics: The $\sigma$-algebra of all Borel sets $A\in\Bor(\Rl)$ forms a logic (a particular orthocomplemented and orthomodular lattice, see Definition~\ref{def:orthocomplementation}), and the same is true for the lattice $\CP(\Hil)$ of orthogonal projections on the Hilbert space~$\Hil$ of the system under consideration, albeit $\CP(\Hil)$ is non-Boolean. Via the spectral theorem, observables may be identified with projection-valued measures (logic homomorphisms) $P:\Bor(\Rl)\to\CP(\Hil)$, and number-valued probability measures $\mu:\CP(\Hil)\to[0,1]$ may be identified with states on $B(\Hil)$ via Gleason's Theorem. Composition $\mu\circ P$ of two such maps then yields the probability measures underlying the Born rule \cite{BirkhoffNeumann:1936,Varadarajan:2006}.

This perspective gives in particular an accurate description of the most fundamental observable of spatial localization of a particle in terms of detection probabilities. In the interest of a more flexible formalism, one may also replace the Borel $\sigma$-algebra $\Bor(\Rl)$ by a more general logic~$\CL$, like for instance a $\sigma$-algebra on $\Rl^d$ of general spatial dimension $d$, or a different configuration space manifold. This setup can naturally be extended to include covariant transformation behaviour under a group of spatial symmetries (such as the Euclidean group $\Eu(d)$), and includes in particular the Newton-Wigner localization observables as systems of imprimitivity for the Euclidean group \cite{NewtonWigner:1949,Wightman:1962}. With causal complementation (non-timelike separation) as the orthocomplementation, also the family $\CC(M)$ of all causally complete subsets of a globally hyperbolic spacetime~$M$ forms a (non-Boolean) logic \cite{Casini:2002,CeglaJancewiczFlorek:2017}. Hence the lattice-theoretic approach also applies to space-time localization instead of spatial localization (although we do not enter a detailed discussion of their interpretation), and related questions even appear in holographic scenarios \cite{LeutheusserLiu:2025_2}.

However, well-known No-Go theorems like Hegerfeldt's Theorem \cite{Hegerfeldt:1974} or Malament's Theorem \cite{Malament:1996} (going back to earlier work by Schlieder \cite{Schlieder:1971} and Jancewicz \cite{Jancewicz:1977}) express a serious incompatibility between localization observables taking values in $\CP(\Hil)$ and basic properties of relativity: In the presence of a dynamics with positive Hamiltonian, there exist no such localization observables that are causal in the sense of exhibiting finite speed of propagation or being compatible with spacelike (or non-timelike) separation in Minkowski spacetime. This conflict rules out both, spatial and spacetime localization observables. In particular, the projection-valued measures of Newton and Wigner, despite all their desirable properties of localization observables \cite{Wightman:1962,Moretti:2023}, violate causality.

These causality problems disappear when moving on to the richer framework of quantum field theory. In the operator-algebraic approach to QFT, one considers instead of the projection lattice $\CP(\Hil)$ a lattice of local von Neumann algebras $\CA(\CO)\subset B(\Hil)$ indexed by causally complete regions~$\CO$ \cite{Haag:1996}. The localization of a particle in some region of space or spacetime is then tested by a comparison of expectation values of local observables \cite{BuchholzYngvason:1994} in the single particle and vacuum states, and is perfectly consistent with causality. In a vacuum representation with vacuum vector $\Om$, the lattice of local von Neumann algebras defines in particular a lattice of closed real subspaces $\{A\Om\col A=A^*\in\CA(\CO)\}^{\norm{\cdot}}\subset\Hil$. These spaces provide a quantum-field theoretic localization structure known as {\em modular localization} \cite{BrunettiGuidoLongo:2002,Schroer:1997-1}, which is decidedly different from the quantum-mechanical Newton-Wigner localization concept \cite{Halvorson:2001}.

While there are no causality problems in quantum field theory, interest in relativistic localization observables in quantum mechanics persists (see, for example, \cite{CastrigianoMutze:1982,PalmerTakahashiWestman:2012,Terno:2014,CastrigianoLeiseifer:2015,Castrigiano:2017,CeglaJancewiczFlorek:2017}, for a small sample of a large body of literature, containing in particular the extensive work of Castrigiano on the subject). In order to avoid the mentioned No-Go results, often the larger family of effects $\CE(\Hil)=\{X\in B(\Hil)\col0\leq X\leq1\}\supset\CP(\Hil)$ instead of $\CP(\Hil)$ is considered, making use of positive-operator valued measures (POVMs) instead of projection-valued measures. A version of Schlieder's Theorem implies that these POVMs can not be commutative \cite{Busch:1999} if they are causal, i.e. they must correspond to unsharp measurements \cite{BuschLahtiMittelstaedt:1996}. Non-commutative POVMs avoiding the conflict with causality do exist. Starting from an example of Terno \cite{Terno:2014} which is based on the one-particle projection of an energy-momentum tensor, Moretti found a causal spatial localization observable for massive Klein-Gordon particles that recovers the Newton-Wigner operators as its first moment \cite{Moretti:2023}. These observables have recently been generalized to a large family parameterized by conserved currents and causal kernels by Castrigiano, De Rosa and Moretti \cite{DeRosaMoretti:2024,CastrigianoDeRosaMoretti:2025}. However, the logic and lattice properties of $\CP(\Hil)$ get lost when generalizing to $\CE(\Hil)$.

\medskip

In this article, we pursue an alternative approach by considering quantum-mechanical localization observables and their causality properties from a lattice-theoretic perspective. We will not work in the framework of quantum field theory, but formulate localization observables in terms of suitable maps $\CL\to\CQ$ from a logic~$\CL$, to be thought of as a $\sigma$-algebra for spatial localization or the logic of causally complete regions in Minkowski space for spacetime localization, to a lattice~$\CQ$ generalizing $\CP(\Hil)$. It will be important to go beyond the setting of orthomodular lattices and work with the more general involution lattices instead. The main feature of an involution lattice is that it is equipped with an order-reversing involution $H\mapsto H^\perp$ which does not have to be a complement, i.e. $H\wedge H^\perp=0$ and $H\vee H^\perp=1$ may fail (Definition~\ref{def:involutionlattice}).

A prominent example of this is the lattice $\CP_{\Rl}(\Hil)$ of all {\em real-linear} closed subspaces (or, equivalently, real linear projections) $H$ of a {\em complex} Hilbert space $\Hil$. In this case, the involution is the symplectic complementation
\begin{align}
    H^\perp=\{\psi\in\Hil\col\imag\langle\psi,h\rangle=0\,\forall h\in H\},
\end{align}
familiar from the canonical commutation relations. The lattice $\CP_{\Rl}(\Hil)$ contains the logic $\CP(\Hil)$ as a sublogic but is not orthocomplemented itself (hence, in order to avoid confusion, we use the symbol $H'$ instead of $H^\perp$ in the body of the article). As mentioned before, examples of projections in $\CP_{\Rl}(\Hil)$ can be obtained from local von Neumann algebras, but no connection to QFT needs to be assumed when studying $\CP_{\Rl}(\Hil)$.

The basic lattice-theoretic notions underlying our approach are presented in Section~\ref{section:Lattices}, where also several examples of logics and involution lattices (including besides $\CP_{\Rl}(\Hil)$ a remark how the set of effects $\CE(\Hil)$ is also an involution lattice when equipped with spectral order, and the involution lattice of von Neumann algebras on $\Hil$) are discussed.

Localization observables are defined in Section~\ref{section:LocalizationObservables} as maps $\CL\to\CQ$ that obey a normalization condition, are $\sigma$-additive, and preserve lattice-theoretic separation (Definition~\ref{def:Observable}). In this setup, a natural notion of dynamics and causality can be introduced, reminiscent of Castrigiano's causality condition \cite{Castrigiano:2017}. We review how causal localization observables mapping into $\CP(\Hil)$ are ruled out by a version of Malament's Theorem (Theorem~\ref{thm:NoGo}), and show that the constraints imposed on localization observables mapping into $\CP_{\Rl}(\Hil)$ are weaker but non-trivial. In particular, the existence of a causal translationally covariant spatial localization observable is shown to imply Lorentz invariance if a strong version of causality is imposed, at least at the level of the energy-momentum spectrum (Proposition~\ref{proposition:SpectrumMustBeLorentzInvariant}). Another important point is that a causal localization observable in $\CP_{\Rl}(\Hil)$ automatically leads to modular localization in terms of standard subspaces (see \cite{Longo:2008_2} and the upcoming textbook \cite{CorreadaSilvaLechnerLongo:2026}, and Appendix~\ref{appendix:StandardSubspaces} for a minimal review).

In Section~\ref{section:SpacetimeLocObinPRH}, we consider the most constrained setting of spacetime localization observables in $\CP_{\Rl}(\Hil)$ that transform covariantly under a positive energy representation $U$ of the Poincaré group $\Poi(d+1)$. We restrict ourselves to massive representations and consider the Brunetti-Guido-Longo (BGL) map \cite{BrunettiGuidoLongo:2002} which defines a family of real projections in terms of the representation $U$. This map was originally formulated as a quantization-free approach to defining free quantum field theories, has been used as a starting point for the construction of interacting QFT models \cite{CorreaDaSilvaLechner:2023}, and has been generalized to other Lie groups \cite{MorinelliNeeb:2023}. Here we explain how the BGL map is an example of a spacetime localization observable. As spacetime localization observables are automatically causal, and can be restricted to causal spatial localization observables (Proposition~\ref{prop:SpacetimeObservableInduceCausalSpatialObservables}), this shows in particular that causal spatial localization observables exist in $\CP_{\Rl}(\Hil)$. We also discuss to which extent the BGL map is the unique spacetime localization observable for a given representation $U$ (see Section~\ref{section:uniqueness}).

In Section~\ref{section:Probability} we return to the Born rule. Guided by Gleason's Theorem for $\CP(\Hil)$, we define probability measures on involution lattices and prove a Gleason-type theorem (Theorem~\ref{theorem:SymplecticGleasonTheoremHilbert}, proved in the setting of symplectic spaces in Appendix~\ref{section:SymplecticGleason}) for~$\CP_{\Rl}(\Hil)$. In contrast to the original Gleason Theorem, it says that in the case of an infinite-dimensional separable Hilbert space, {\em no} probability measures on $\CP_{\Rl}(\Hil)$ exist. In line with this result we find that evaluation of a spacetime localization observable in $\CP_{\Rl}(\Hil)$ in (the real part of) a state gives only a fuzzy probability measure; the additivity property fails. However, we also adapt Fredenhagen's cluster theorem \cite{Fredenhagen:1985} to standard subspaces (Theorem~\ref{thm:MassGapClustering}) and show that given a causal spatial localization observable $E:\Bor(\Rl^d)\to\CP_{\Rl}(\Hil)$, the maps $\mu^E_\om:A\mapsto\real\om(E(A))$ (with $\om$ a normal state on $B(\Hil)$) are approximately additive in the sense that
\begin{align}
    |\mu^E_\om(A\vee B)-\mu^E_\om(A)-\mu^E_\om(B)|
    \leq
    e^{-m\cdot d(A,B)}
\end{align}
holds for arbitrary spatial Borel regions $A,B$ (with $m$ the mass of the particle and $d$ the Euclidean distance). In a similar spirit, one can also show that a vector $\psi=E(B)\psi$ that is localized in the spatial region $B$ according to the causal spatial localization observable $E:\Bor(\Rl^d)\to\CP_{\Rl}(\Hil)$ is also essentially Newton-Wigner localized in a slightly larger region $B_\delta$. These facts show that the causal localization observables fail to produce genuine measures, but the error in additivity is negligible for macroscopic configurations.

\medskip

We conclude that the two extensions of the projection logic $\CP(\Hil)$ by real projections and (complex) effects,
\begin{align}
    \CP_{\Rl}(\Hil) \supset \CP(\Hil) \subset \CE(\Hil),
\end{align}
lead to quite different results. The main message of this article is that when adopting the point of view that a (spatial or spacetime) localization observable is given by map between a logic of localization regions and the involution lattice $\CP_{\Rl}(\Hil)$, one arrives at a fully relativistic and causal description of localization, and in fact several aspects of quantum field theory (such as Lorentz symmetry and modular localization) emerge automatically. While this approach leads only to an approximate description of localization in terms of probability measures, it has the advantage of providing an essentially unique localization observable and a clear lattice structure. In contrast, the approach based on POVMs and effects yields a clear measure-theoretic picture, but seems to have less similarity with quantum field theory and lattice notions.

A detailed comparison of these two ideas would be very interesting, in particular given the link of Terno's localization observable to the energy momentum tensor of a free quantum field theory, but has to be left for a future work. Another aspect of importance that is not treated in this paper is an incorporation of the measurement process, as recently generalized from quantum mechanics \cite{BuschLahtiMittelstaedt:1996} to quantum field theory \cite{FewsterVerch:2020}.

\medskip

This article is written from a quantum-mechanical perspective using lattices. As we proceed, several notions of (algebraic) quantum field theory will appear naturally (such as modular localization, the Reeh-Schlieder property, and the Bisognano-Wichmann property), but no previous knowledge in these subjects is assumed. Some of our results are also of independent interest for readers primarily interested in standard subspaces, such as our Gleason Theorem (Theorem~\ref{theorem:SymplecticGleasonTheoremHilbert}) or the standard subspace cluster theorem (Theorem~\ref{thm:MassGapClustering}).

\section{Involution lattices and logics}\label{section:Lattices}

In this section, we describe our general framework of logics and involution lattices. A detailed presentation of the theory of abstract lattices and their links to quantum theory and logic can be found in \cite{Gratzer:2011} and \cite{Redei:1998},\cite[App.~D]{Landsman:2017}, respectively. For us it will be sufficient to review the basic definitions and introduce some important generalizations of them to discuss the examples relevant to our subsequent study of localization observables.

\subsection{Basic definitions}

A {\em lattice} is a non-empty set $\CL$ with two binary, idempotent, commutative, and associative operations $\wedge,\vee$ satisfying the absorption laws
\begin{align}
    A\vee (A\wedge B)=A,\qquad A\wedge(A\vee B)=A,\qquad\quad A,B\in\CL.
\end{align}
Equivalently, a lattice is a partially ordered set $(\CL,\leq)$ admitting a supremum (join) $A\vee B$ and infimum (meet) $A\wedge B$ w.r.t. $\leq$ for any two $A,B\in \CL$. The partial order can then be defined as $A\leq B:\Leftrightarrow A\wedge B=A$ in terms of $\wedge$. In case a lattice admits a largest and smallest element w.r.t. $\leq$, these are unique and denoted $1$ and $0$, respectively. We speak of a {\em bounded} lattice in this case. An element $B\in\CL$ is called {\em complement} of $A\in\CL$ if $A\vee B=1$ and $A\wedge B=0$.

\begin{definition}\label{def:orthocomplementation}
    Let $(\CL,\wedge,\vee)$ be a bounded lattice.
    \begin{enumerate}
     \item\label{item:orthocomplement} An {\em orthocomplementation} on $\mathcal{L}$ is a map $\perp:\CL\to\CL$, $A\mapsto A^\perp$ such that
    \begin{enumerate}[{\textrm (i)}]
        \item\label{item:involution} $A^{\perp\perp}=A$,
        \item\label{item:orderreversing} $A\leq B$ implies $B^\perp\leq A^\perp$,
        \item\label{item:complement} $A^\perp$ is a complement of $A$.
    \end{enumerate}

    \item $\CL$ is {\em $\sigma$-complete} if for any countable subset $\CL_0\subset\CL$, the supremum $\bigvee_{A\in\CL_0}A$ and infimum $\bigwedge_{A\in\CL_0}A$ exist in $\CL$, and {\em complete} if $\bigvee_{A\in\CL_0}A$ and $\bigwedge_{A\in\CL_0}A$ exist even for arbitrary subsets $\CL_0\subset\CL$.

    \item $\CL$ is {\em distributive} if
    \begin{align}
        A\wedge (B\vee C)
        =
        (A\wedge B)\vee (A\wedge C),\qquad A,B,C\in\CL,
    \end{align}
    and {\em Boolean} if it is distributive and orthocomplemented.

    \item\label{item:LatticeOrthomodular} An {\em orthomodular lattice} is an orthocomplemented lattice $\CL$ such that
    \begin{align}
        A\leq B \Rightarrow B=A\vee(A^\perp\wedge B),\qquad A,B\in\CL.
    \end{align}

    \item A {\em logic} is a $\sigma$-complete orthomodular lattice.
    \end{enumerate}
\end{definition}

The propositions of classical logic form an example of a Boolean logic, with the partial order given by implication, the orthocomplement by negation, and $1$ the tautology. The lattice operations join/meet are then given by logical or/and. As any Boolean lattice, it is orthomodular. In contrast, {\em quantum} logic, formalized as the lattice of orthogonal projections in a Hilbert space (see Example~\ref{example:PH}), is an example of a non-Boolean logic.

While Definition~\ref{def:orthocomplementation} is standard, we will also use a more general concept that is introduced now.

\begin{definition}\label{def:involutionlattice}
    An {\em involution lattice} is a $\sigma$-complete bounded lattice with an order-reversing involution $\perp$ (satisfying~\ref{item:involution} and \ref{item:orderreversing} but not necessarily \ref{item:complement} in Definition~\ref{def:orthocomplementation}~\ref{item:orthocomplement}).
\end{definition}

Our terminology is set up in such a way that boundedness is always assumed, and logics or involution lattices without further adjectives are $\sigma$-complete. In the case of unrestricted completeness as defined above, we speak of complete logics or involution lattices.

In an involution lattice, de Morgan's laws
\begin{align}\label{eq:DeMorgansLaws}
    (A\vee B)^\perp=A^\perp\wedge B^\perp,\qquad
    (A\wedge B)^\perp=A^\perp\vee B^\perp,
\end{align}
still hold (their standard proof does not use that $\perp$ is an orthocomplementation, see for example \cite[Proposition~3.4]{Redei:1998}), but $A^\perp$ may fail\footnote{In terms of logical statements, the law of excluded middle $A\vee A^\perp=1$ may be violated in an involution lattice, which may be seen as a desirable feature of a quantum logic \cite{Landsman:2017}.} to be a complement of $A$.
Involution lattices are used in the investigation of unsharp quantum logics \cite{ChiaraGiuntiniGreechie:2004}. We will see examples of such lattices below.

Involution lattices in which $\perp$ is not an orthocomplementation are not orthomodular; this follows by setting $B=1$ in the orthomodularity condition Definition~\ref{def:orthocomplementation}~\ref{item:LatticeOrthomodular}.

\medskip

In our subsequent analysis of localization observables in Section~\ref{section:LocalizationObservables}, covariance and causality properties will be at the center of our interest. A discussion of these properties requires two more concepts for involution lattices that we introduce now.

\subsection*{Separation and disjointness in involution lattices}

When discussing causality questions, we will need to distinguish between separated and disjoint elements.

\begin{definition}\label{def:SeparationAndDisjointness}
    Let $\CL$ be an involution lattice.
    \begin{enumerate}
        \item $A,B\in\CL$ are called {\em separated} if $A\leq B^\perp$. For orthocomplemented $\CL$, we also say that $A$ and $B$ are {\em orthogonal} in this case.

        \item $A,B\in\CL$ are called {\em disjoint} if $A\wedge B=0$.
    \end{enumerate}
\end{definition}

Note that in view of the properties of the involution map $\perp$, the symmetric formulation is justified: $A$ is separated from $B$ if and only if $B$ is separated from $A$.

It is easy to see that in a Boolean lattice, two elements are separated if and only if they are disjoint. In general orthocomplemented lattices, only the implication ``separated''$\Rightarrow$``disjoint'' holds\footnote{In the logic of closed subspaces of a Hilbert space (Example~\ref{example:PH}), consider two non-orthogonal subspaces that intersect trivially to get an example of two elements that are disjoint but not orthogonal.}. In general involution lattices, separation and disjointness are logically independent and there may even exist elements $A\neq0$ that are separated from themselves, $A\leq A^\perp$ (see Example~\ref{example:PRH}).

As we will work with (generalizations of) the non-Boolean logic of quantum mechanics, we will use separation of two elements to express their commensurability (e.g., commuting projections as in Example~\ref{example:PH} or causal separation of causally complete regions in Minkowski space as in Example~\ref{example:MinkowskiLogic}), rather than disjointness. We will also say that a subset $\CL_0\subset\CL$ of an involution lattice is {\em separated} if any two distinct elements of $\CL_0$ are separated.

\subsection*{Homomorphisms of involution lattices} Homomorphisms between two involution lattices $\CL_1$, $\CL_2$ are defined as the maps $\varphi:\CL_1\to\CL_2$ preserving the relevant structures, namely $\varphi(0)=0$, $\varphi(1)=1$, $\varphi(A^\perp)=\varphi(A)^\perp$ as well as $\varphi(A\vee B)=\varphi(A)\vee\varphi(B)$ and $\varphi(A\wedge B)=\varphi(A)\wedge\varphi(B)$ for all $A,B\in\CL_1$. A map $\varphi:\CL_1\to\CL_2$ is called {\em monotone} if $A\leq B\Rightarrow\varphi(A)\leq\varphi(B)$, and {\em finitely additive} (respectively {\em $\sigma$-additive} or {\em completely additive}) if
\begin{align}\label{eq:additivity}
    \CL_0\subset\CL\;\text{ separated}
    \Rightarrow
    \varphi\left(\bigvee_{A\in\CL_0}A\right)
    =
    \bigvee_{A\in\CL_0}\varphi\left(A\right)
\end{align}
holds for all finite (respectively countable, or arbitrary) separated subsets $\CL_0\subset\CL$. Homomorphisms are clearly monotone and finitely additive.

The concept of a homomorphism of course restricts to the special cases of logics and Boolean logics. Iso- and automorphisms are defined in the canonical way. In particular, any involution lattice~$\CL$ has an automorphism group $\Aut\CL$.

\begin{definition}\label{def:AutomorphicActionsAndCovariance}
    Let $G$ be a group and $\CL_1$, $\CL_2$ involution lattices.
    \begin{enumerate}
        \item\label{item:AutomorphicAction} An {\em automorphic action} of $G$ on $\CL_1$ is a group homomorphism $\alpha:G\to\Aut\CL_1$.

        \item\label{item:Covariance} Given two automorphic actions $\alpha_j:G\to\Aut\CL_j$, $j=1,2$, a map $f:\CL_1\to\CL_2$ is {\em covariant} if $f\circ\alpha_{1,g}=\alpha_{2,g}\circ f$ for all $g\in G$.
    \end{enumerate}
\end{definition}

\medskip

We now discuss various examples of logics and involution lattices that are relevant for localization observables.

\subsection{Space- and spacetime logics}

The most basic geometric examples of logics come from measure theory: Every $\sigma$-algebra on a set $X$ is a Boolean logic with partial order given by inclusion, orthocomplementation given by set-theoretic complement $A\mapsto A^c$, and $0=\emptyset$, $1=X$. In this case, $A\vee B=A\cup B$ and $A\wedge B=A\cap B$ for $A,B$ in the $\sigma$-algebra.

\begin{example}[\bfseries Spatial localization regions, the Boolean logic of a Cauchy surface]\label{example:SpatialLocalization}
    The $\sigma$-algebra $\Bor(\Rl^d)$ of all Borel subsets of Euclidean space~$\Rl^d$ is a Boolean logic with its set-theoretic operations. The Euclidean group $\Eu(d)$ acts on it by automorphisms according to $\alpha_g(A):=g(A)$.
\end{example}

$\Bor(\Rl^d)$ and its Euclidean action will be used as a model for spatial localization regions. Instead of $\Bor(\Rl^d)$, one could also work with the larger $\sigma$-algebra $\Meas(\Rl^d)$ of all Lebesgue measurable sets\footnote{In certain situations, this change from $\Bor(\Rl^d)$ to $\Meas(\Rl^d)$ is necessary: In our later discussion of dynamics, we will in particular consider ``smeared out'' regions of the form $A+|t|\bB=\bigcup_{x\in A}(x+|t|\bB)$, $t\neq0$, with $\bB$ the  unit ball in $\Rl^d$. If $\bB$ is the open unit ball, this set is open, so in particular Borel. However, if $\bB$ is closed and $d>1$, there exist Borel sets $A$ such that $A+|t|\bB$ is not Borel, whereas it is Lebesgue measurable for any set $A\subset\Rl^d$ (\cite[Exercise~7.20]{Rudin:1970}, see also \cite[Lemma~16]{Castrigiano:2017} for a related statement).}.

\bigskip

We next consider Minkowski space{\em time} $\Rl^{d+1}$ (with points $x=(x_0,\bx)\in\Rl\times\Rl^d$, Lorentz inner product $x\cdot y=x_0y_0-\bx\by$ and $x^2:=x\cdot x$). To define an appropriately complemented lattice of regions in $\Rl^{d+1}$, the physical idea is to consider two points $x,y\in\Rl^{d+1}$ as separated when no signal can be sent from one to the other. Depending on whether signals propagating at the speed of light are considered or not, one considers the relation that $x,y$ are {\em spacelike}, written $x\sim_{\mathrm s}y$, if $(x-y)^2<0$, or {\em non-timelike} \cite[Sect.~24]{Castrigiano:2017} ({\em achronal} \cite{CastrigianoDeRosaMoretti:2025}), written $x\sim y$, defined by
\begin{align}\label{eq:CausalSeparationOfPoints}
    x\sim y &:\Leftrightarrow x\neq y\;\text{and}\,(x-y)^2\leq0.
\end{align}
The relation $\sim$ considers points $x,y\in\Rl^{d+1}$ with $x-y$ lightlike to be separated, as appropriate for massive particles which move at speeds smaller than the speed of light. As we will not discuss massless particles in this article, we use $\sim$ instead of $\sim_{\mathrm s}$, which leads to an improved logic structure (see below).

Based on the separation relation $\sim$, we define the {\em causal complement} of any set $\CO\subset\Rl^{d+1}$ as
\begin{align}\label{eq:causalcomplement}
    \CO':= \{x\in \Rl^{d+1}\,:\, x\sim y\;\forall y\in \CO\}.
\end{align}
With $V^\pm\subset\Rl^{d+1}$ the open forward/backward lightcone, the chronological future/past $I^\pm(\CO)$ of a set $\CO\subset\Rl^{d+1}$ is
\begin{align}\label{eq:defI}
    I^\pm(\CO)
    :=
    \CO+V^\pm,\qquad
    I(\CO)
    :=
    \CO\cup I^+(\CO)\cup I^-(\CO),
\end{align}
and we have $\CO'=\Rl^{d+1}\setminus I(\CO)$. Hence the causal complement \eqref{eq:causalcomplement} is naturally related to the chronological future/past, instead of the causal future/past (based on the {\em closed} light cone). As $I^\pm(\CO)$ is open, it is clear that $I(\CO)$ is Borel if $\CO$ is Borel.

We will call $\CO$ {\em causally complete} if $\CO=\CO''$. To arrive at a lattice in which the causal complement is an orthocomplementation, we clearly need to restrict to the set $\CC(\Rl^{d+1})$ of all causally complete subsets of $\Rl^{d+1}$.

\begin{example}[{\bfseries Spacetime localization regions, the logic of Minkowski spacetime}]\label{example:MinkowskiLogic}
    The collection $\CC(\Rl^{d+1})$ of all causally complete subsets of $\Rl^{d+1}$ forms a non-Boolean complete logic with partial order given by inclusion, orthocomplementation given by the causal complement $\CO\mapsto\CO'$ \eqref{eq:causalcomplement}, and $0=\emptyset$, $1=\Rl^{d+1}$. In this case,
    \begin{align}\label{eq:meetjoinMinkowski}
        \CO_1\vee \CO_2 &:= \left(\CO_1\cup \CO_2\right)'',\qquad
        \CO_1\wedge \CO_2:= \CO_1\cap \CO_2.
    \end{align}
    The Poincaré group $\Poi(d+1)$ of\, $\Rl^{d+1}$ acts automorphically\footnote{A theorem of Borchers and Hegerfeldt implies that the full automorphism group of $\CC(\Rl^{d+1})$ is generated Poincaré transformations and dilations \cite{BorchersHegerfeldt:1972}.} on $\CC(\Rl^{d+1})$ by $\alpha_{(x,\La)}(\CO)=\La\CO+x$.
\end{example}

A proof of these claims can be found in \cite{CeglaJadczyk:1977}, see also \cite{Casini:2002,CeglaJancewiczFlorek:2017} for an extension to curved (globally hyperbolic) spacetimes. Figure~\ref{figure:CausalLogic} illustrates the difference between disjointness and separation in $\CC(\Rl^{d+1})$, and the failure of distributivity. A subtle point concerns the orthomodularity of $\CC(\Rl^{d+1})$, which {\em fails} if the non-timelike separation relation $\sim$ is replaced by the spacelike separation relation $\sim_{\mathrm s}$ \cite{Casini:2002}. Hence $\CC(\Rl^{d+1})$ is a logic only because we use the non-timelike separation relation.

The logic $\CC(\Rl^{d+1})$ and its Poincaré action will be used as a model for space{\em time} localization regions. Typical elements of $\CC(\Rl^{d+1})$ are double cones (up to a set of measure zero, an intersection of a forward and backward lightcone), the closed wedge $W:=\{x\in\Rl^{d+1}\col x_1\geq|x_0|\}$, one-point sets $\{x\}$, and sets constructed from these examples by join, meet, and Poincaré transformations. All sets in $\CC(\Rl^{d+1})$ are Lebesgue measurable \cite[Sect.~24.4]{Castrigiano:2017}.

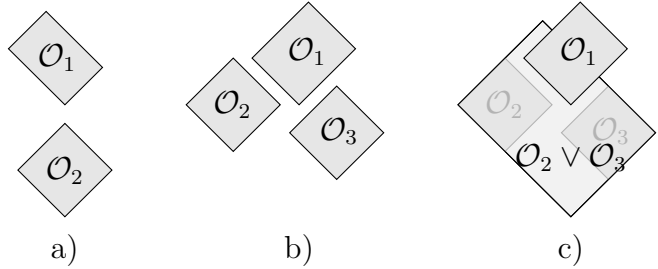
\begin{figure}[h]
    \centering
    \begin{tikzpicture}[scale=0.7]
        \coordinate (center) at (0,0);
        \begin{scope}[rotate around={45:(center)}]
            \draw[draw=black,fill=black!10!white] (-1.5,-1.5) rectangle (-0.25,-0.25);
            \draw[draw=black,fill=black!10!white] (0,0) rectangle (1,1.5);
        \end{scope}
        \node at (-0.15,0.9) {$\CO_1$};
        \node at (0,-1.25) {$\CO_2$};
        \node at (0,-2.75) {a)};
    \end{tikzpicture}
\qquad
    \begin{tikzpicture}[scale=0.7]
    \coordinate (center) at (0,0);
    \begin{scope}[rotate around={-45:(center)}]
        \draw[draw=black,fill=black!10!white] (-1.5,-1.5) rectangle (-0.25,-0.25);
        \draw[draw=black,fill=black!10!white] (0.25,-0.5) rectangle (1.5,0.75);
        \draw[draw=black,fill=black!10!white] (-1.25,0) rectangle (0,1.5);
    \end{scope}
    \node at (0.15,1) {$\CO_1$};
    \node at (-1.25,0) {$\CO_2$};
    \node at (0.75,-0.5) {$\CO_3$};
    \node at (0,-2.75) {b)};
    \end{tikzpicture}
\qquad
    \begin{tikzpicture}[scale=0.7]
    \coordinate (center) at (0,0);
    \begin{scope}[rotate around={-45:(center)}]
        \draw[draw=black,fill=black!5!white] (-1.5,-1.5) rectangle (1.5,0.75);
        \draw[draw=black!30!white,fill=black!10!white] (-1.5,-1.5) rectangle (-0.25,-0.25);
        \draw[draw=black!30!white,fill=black!10!white] (0.25,-0.5) rectangle (1.5,0.75);
        \draw[draw=black,fill=none] (-1.5,-1.5) rectangle (1.5,0.75);
        \draw[draw=black,fill=black!10!white] (-1.25,0) rectangle (0,1.5);
    \end{scope}
    \node at (0.15,1) {$\CO_1$};
    \node at (-1.25,0) {\textcolor{black!30!white}{$\CO_2$}};
    \node at (0.75,-0.5) {\textcolor{black!30!white}{$\CO_3$}};
    \node at (0,-1) {$\CO_2\vee\CO_3$};
    \node at (0,-2.75) {c)};
    \end{tikzpicture}
    \caption{a) $\CO_1$ and $\CO_2$ are disjoint but not separated. b),c) Failure of distributivity: $\CO_1\wedge\CO_2=\emptyset=\CO_1\wedge\CO_3$ (b), but $\CO_1\wedge(\CO_2\vee\CO_3)\neq\emptyset$~(c). In these figures, the time axis is vertical and the spatial axis horizontal as usual.}\label{figure:CausalLogic}
\end{figure}

From the point of view of lattice structures, the essential difference between the logic $\Bor(\Rl^d)$ describing localization regions in Euclidean space and the logic $\CC(\Rl^{d+1})$ describing localization regions in Minkowski spacetime is that the former is Boolean but the latter is not.

We may of course think of $\Rl^d$ as a Cauchy surface $\{0\}\times\Rl^d$ in Minkowski spacetime $\Rl\times\Rl^d$. This leads to a covariant embedding by causal completion.

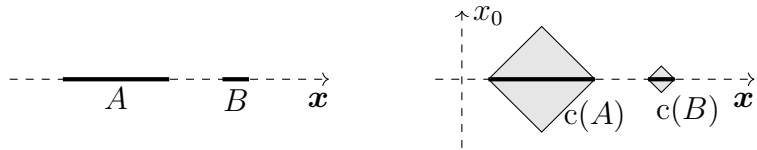
\begin{figure}[h]
    \centering
    \begin{tikzpicture}[scale=0.7]
        \draw[->,dashed] (-3,0) to (3,0);
        \draw[ultra thick] (-2,0) to (0,0);
        \draw[ultra thick] (1,0) to (1.5,0);
        \node at (-1,-0.35) {$A$};
        \node at (1.25,-0.4) {$B$};
        \node at (2.8,-0.4) {$\bx$};
        \begin{scope}[shift={(8,0)}]
            \draw[->,dashed] (-3,0) to (3,0); 
            \draw[->,dashed] (-2.5,-1.3) to (-2.5,1.3); 
            \draw[draw=black,fill=black!10!white] (-2,0) -- (-1,1) -- (0,0) -- (-1,-1) --cycle;
            \draw[draw=black,fill=black!10!white] (1,0) -- (1.25,0.25) -- (1.5,0) -- (1.25,-0.25) --cycle;
            \draw[ultra thick] (-2,0) to (0,0);
            \draw[ultra thick] (1,0) to (1.5,0);
            \node at (0,-.7) {$\cc(A)$};
            \node at (1.75,-.6) {$\cc(B)$};
            \node at (2.8,-0.4) {$\bx$};
            \node at (-2,1.2) {$x_0$};
        \end{scope}
    \end{tikzpicture}
\caption{The causal completion map $\cc(A)=(\{0\}\times A)''$.}
\end{figure}

\begin{lemma}\label{lemma:CausalEmbedding}
    The causal completion map
    \begin{align}\label{def:c}
        \cc:\Bor(\Rl^d) \to\CC(\Rl^{d+1}),\qquad \cc(A) := (\{0\}\times A)''
    \end{align}
    is an injective $\sigma$-additive logic homomorphism, making $\cc(\Bor(\Rl^d))$ a Boolean sublogic of $\CC(\Rl^{d+1})$. It is covariant w.r.t. the actions of the Euclidean group $\Eu(d)$ on $\Bor(\Rl^d)$ and $\CC(\Rl^{d+1})$.
\end{lemma}
\begin{proof}
    It is clear that $\cc$ is well-defined (maps into $\CC(\Rl^{d+1})$), and we recognize it to be injective because $A$ can be recovered from $\cc(A)$ by $\{0\}\times A=\cc(A)\cap\Sigma$, with $\Sigma:=\{0\}\times\Rl^d$. It is also clear that $\cc$ is normalized in the sense of preserving minimal and maximal elements, i.e. $\cc(\emptyset)=\emptyset$ and $\cc(\Rl^d)=\Rl^{d+1}$.

    Given two disjoint sets $A, B\in\Bor(\Rl^d)$, we have $\{0\}\times A\subset(\{0\}\times B)'$, and hence $\cc(A)\subset\cc(B)'$. Furthermore, for a separated countable family $(A_n)_{n\in\Nl}\subset\Bor(\Rl^d)$, one obtains $\cc(\bigcup_n A_n)=(\bigcap_n(\{0\}\times A_n)')'=\bigvee_n\cc(A_n)$, where we have used that $\CC(\Rl^{d+1})$ is a complete lattice. Thus $\cc$ is $\sigma$-additive.

    The properties of $\cc$ established so far already imply that it is a logic homomorphism (see Lemma~\ref{lemma:Observables}). Alternatively, the proof can be carried out by arguing more directly with the form of the achronal complementation \cite[Lemma~119]{Castrigiano:2017}.
\end{proof}

\subsection{Logics and involution lattices on a Hilbert space}

Quantum logics formalized in terms of closed subspaces of a Hilbert or the corresponding orthogonal projections go back to early work by Birkhoff and von Neumann \cite{BirkhoffNeumann:1936}. As is well known, the set of all closed subspaces of a Hilbert space $\Hil$ with partial order given by inclusion, orthocomplementation given by orthogonal complements, and $0=\{0\}$,  $1=\Hil$, is a complete logic that is non-Boolean for $\dim\Hil>1$. The resulting lattice operations are
\begin{align}\label{eq:meetjoinspaces}
    \CK_1\vee\CK_2=\overline{\text{span}\{\CK_1,\CK_2\}},\qquad \CK_1\wedge\CK_2=\CK_1\cap\CK_2.
\end{align}
We phrase this structure in terms of the orthogonal projections onto closed subspaces, and write $\AU(\Hil)$ for the group of all unitary and antiunitary operators on $\Hil$.

\begin{example}[\bfseries Orthogonal projections, the logic of quantum mechanics]\label{example:PH}
    The set $\CP(\Hil)$ of all orthogonal projections $P\in B(\Hil)$ with partial order given by operator order, orthocomplementation given by $P^\perp=1-P$, and the minimal and maximal projections $0$ and $1$, is a complete logic. The resulting lattice operations are
    \begin{align}
        P\wedge Q = \slim_{n\to\infty}(PQ)^n,
        \qquad
        P\vee Q = 1-\slim_{n\to\infty}(P^\perp Q^\perp)^n.
    \end{align}
    Two projections are separated, i.e. $P\leq Q^\perp$, if and only if $PQ=0$. The group $\AU(\Hil)$ acts automorphically on $\CP(\Hil)$ via $\alpha_U(P)=UPU^*$.\hfill$\square$
\end{example}

Orthogonal projections are commonly used to model observables and sharp measurements in quantum mechanics (see also Example~\ref{example:Effects} for unsharp measurements).

\bigskip

We now define an involution lattice that strictly contains $\CP(\Hil)$ and is of central importance in this paper. Whereas in the previous example, $\Hil$ can be complex or real, it is now important that $\Hil$ is complex. In this case, we may also think of $\Hil$ as a real Hilbert space, with scalar product $\langle\,\cdot\,,\,\cdot\,\rangle_{\Rl}:=\real\langle\,\cdot\,,\,\cdot\,\rangle$ (which induces the same norm as $\langle\,\cdot\,,\,\cdot\,\rangle$), and consider the lattice of all closed real linear subspaces $H\subset\Hil$, with the same lattice operations as before (partial order by inclusion, maximal/minimal element $\Hil$ and $\{0\}$, with resulting meet $H_1\wedge H_2=H_1\cap H_2$ and join $H_1\vee H_2=\overline{\Rl\text{-span}\{H_1,H_2\}}$).

Instead of the orthogonal complement $\perp_{\Rl}$ w.r.t. $\langle\,\cdot\,,\,\cdot\,\rangle_{\Rl}$ we however consider the map
\begin{align}\label{eq:SymplecticComplement}
    H \mapsto H^\prime
    :=
    iH^{\perp_{\Rl}}
    =
    \left\{\xi\in\mathcal{H}\col\imag\langle\xi,h\rangle=0\,\forall h\in H\right\},
\end{align}
mapping $H$ onto its {\em symplectic complement} $H'$ w.r.t. the symplectic form $\imag\langle\,\cdot\,,\,\cdot\,\rangle$. It is clear that this map is order-reversing, and also involutive because
\begin{align}
    H''=i(iH^{\perp_{\Rl}})^{\perp_{\Rl}}=H^{\perp_{\Rl}\perp_{\Rl}}=\overline H=H.
\end{align}
Hence the closed real linear subspaces with symplectic complementation form a complete involution lattice. However, despite its name, the symplectic complement $H'$ of $H$ is often {\em not} a complement: As a simple example, take $H$ to be the closed real span of an orthonormal basis of $\Hil$. Then one easily sees $H=H'\neq\{0\}$, so $H'$ is not a complement of $H$. We conclude that this lattice is not orthocomplemented and in particular not orthomodular, i.e. there exist proper inclusions with trivial relative symplectic complement \cite{CorreadaSilvaLechner:2025}. The same example also shows that spaces that are separated w.r.t. the symplectic complement do not have to be disjoint (cf.~Definition~\ref{def:SeparationAndDisjointness}) in this lattice.

We now formulate this lattice in terms of the real orthogonal (i.e. real linear and selfadjoint w.r.t. the real inner product $\langle\,\cdot\,,\,\cdot\,\rangle_{\Rl}$) projections onto closed real subspaces.

\begin{example}[\bfseries Real orthogonal projections]\label{example:PRH}
    The set $\CP_{\Rl}(\Hil)$ of all real orthogonal projections $E$ on a complex Hilbert space $\Hil$, with partial order given by $E\leq F\Leftrightarrow E\Hil\subset F\Hil$, involution
    \begin{align}
        E':=1+iEi,
    \end{align}
    and the minimal and maximal projections $0$ and $1$, is a non-orthocomplemented complete involution lattice. The resulting lattice operations are
    \begin{align}\label{eq:JoinInPRH}
        E\wedge F = \slim_{n\to\infty}(EF)^n,
        \qquad
        E\vee F = 1+i\slim_{n\to\infty}(E'F')^n\,i.
    \end{align}
    Two real orthogonal projections are separated, i.e. $E\leq F'$, if and only if $EiF=0$. The group $\AU(\Hil)$ acts automorphically on $\CP_{\Rl}(\Hil)$ via $\alpha_U(E)=UEU^*$.
\end{example}
\begin{proof}
    We have already shown above that the lattice of all closed real subspaces with symplectic complement \eqref{eq:SymplecticComplement} is a non-orthocomplemented complete involution lattice, and now only need to transfer that information to the real orthogonal projections onto these spaces. The minimal/maximal element, order, and join for $\CP_{\Rl}(\Hil)$ are simply the real linear versions of the lattice in Example~\ref{example:PH} and hence need no further explanation.

    Taking into account that the real adjoint of a unitary or antiunitary operator $U$ coincides with its complex adjoint, we see that the automorphic $\AU(\Hil)$-action $H\mapsto UH$ on the subspace lattice transfers to $\alpha_U(E_H):=E_{UH}=UE_HU^*$, where $E_H$ denotes the real orthogonal projection onto $H$. By the same argument, the symplectic complement $H'=iH^{\perp_{\Rl}}$ transfers to
    \begin{align*}
        E_{H'}=E_{iH^{\perp_{\Rl}}}=-iE_{H^{\perp_{\Rl}}}i=-i(1-E_H)i=1+iE_H i,
    \end{align*}
    as claimed. By de Morgan's law \eqref{eq:DeMorgansLaws}, this also implies the formula \eqref{eq:JoinInPRH} for the join in $\CP_{\Rl}(\Hil)$.

    It remains to show the characterization of $E\leq F'$. This separation is equivalent to $E=F'E=EF'=E(1+iFi)$, and hence equivalent to $EiF=0$, as claimed.
\end{proof}

\begin{lemma}\label{lemma:CPintoCPR}
    The embedding $\CP(\Hil)\hookrightarrow\CP_{\Rl}(\Hil)$, $P\mapsto P$, is an injective completely additive $\AU(\Hil)$-covariant homomorphism of involution lattices.
\end{lemma}
\begin{proof}
    As any complex orthogonal projection is in particular a real orthogonal projection, we have an (injective) completely additive $\AU(\Hil)$-covariant embedding which is immediately seen to preserve $0$, $1$, and joins. For $P\in\CP(\Hil)$, we have $P'=1+iPi=1-P=P^\perp$ by linearity of $P$, so the embedding also respects the involutions and hence meets.
\end{proof}

For both the lattices $\CP(\Hil)$ and $\CP_{\Rl}(\Hil)$, we will freely switch between the subspace picture and the projection picture depending on what is more convenient. The projection onto a complex subspace $\CK$ and real subspace $H$ will be denoted $P_{\CK}$ and $E_H$, respectively. We will reserve the prime symbol $H'$ (and $E_H':=E_{H'}$) to denote symplectic complements.

\bigskip

To further illustrate the concept of an involution lattice, we give two additional examples that are of relevance in quantum physics, although we will not use them in this paper.

\subsection{Effects and spectral order}\label{section:SpectralOrder} As models for unsharp measurements, often a different extension of~$\CP(\Hil)$ is considered, namely the set of {\em effects} \cite{BuschLahtiMittelstaedt:1996}
\begin{align}
    \CE(\Hil)
    =
    \{X\in\B(\Hil) \col 0\leq X\leq 1\}
    .
\end{align}
It is well known that $\CE(\Hil)$, and more generally the set $B(\Hil)_{\sa}$ of all bounded selfadjoint operators, is not a lattice w.r.t. the usual operator order (see, for example, \cite{Kadison:1951,MorelandGudder:1999}). However, effects can be given the structure of an involution lattice w.r.t. a different order. To that end, recall that the {\em spectral order} $\leqs$ on $\CE(\Hil)$ is defined as \cite{Olson:1971}
\begin{align}\label{eq:SpectralOrder}
    X\leqs Y\,:\Leftrightarrow\, E^X_\la \geq E^Y_\la,\quad \la\in\Rl,
\end{align}
where $E^X_\la:=\chi_{(-\infty,\la]}(X)$ is the right-(strongly) continuous spectral resolution of $X$. In this partial order, $\CE(\Hil)$ is a complete bounded lattice with minimal/maximal elements $0,1$. The spectral order is different from the operator order, but in case $X$ and $Y$ commute, $X\leq Y$ is equivalent to $X\leqs Y$. Moreover, for $X,Y$ effects, $X\leqs Y$ is equivalent to $X^n\leq Y^n$ for all $n\in\Nl$.

The join $\vees$ and meet $\wedges$ of $X,Y\in \CE(\Hil)$ w.r.t. spectral order are the unique selfadjoint effects $X\vees Y$ and $X\wedges Y$ with spectral resolution $E^{X\vees Y}_\la=E^X_\la\wedge E^Y_\la$ and $E^{X\wedges Y}_\la=\bigwedge_{\mu>\la}(E^X_\la\vee E^Y_\la)$, respectively \cite{Olson:1971}.

\begin{example}[\bfseries Effects with spectral order as an involution lattice]\label{example:Effects}
    The set $\CE(\Hil)$ of effects is a complete non-orthocomplemented involution  lattice w.r.t. the spectral order \eqref{eq:SpectralOrder}, the usual minimal/maximal projections $0$ and $1$, and the order-reversing involution
    \begin{align}
        X^\perp = 1-X.
    \end{align}
    Two effects are separated, i.e. $X\leqs Y^\perp$, if and only if $X^n\leq(1-Y)^n$ for all $n\in\Nl$. The group $\AU(\Hil)$ acts automorphically on $\CE(\Hil)$ via $\alpha_U(X)=UXU^*$.
\end{example}

We refer to the literature \cite{Olson:1971,deGroote:2004,Hamhalter:2007} for the proof of these claims. For example, it is proven in \cite[Proposition~4.2]{deGroote:2004} that $X\wedge X^\perp=0$ is equivalent to $X$ being a projection.

Since spectral order $X\leqs Y$ is equivalent to $X^n\leq Y^n$ for all $n$, one can also easily show that the embedding $\CP(\Hil)\hookrightarrow\CE(\Hil)$, $P\mapsto P$, is an injective $\AU(\Hil)$-covariant completely additive homomorphism of involution lattices.

\subsection{The involution lattice of von Neumann algebras}\label{section:vonNeumannAlgebras}
Another important example of an involution lattice is the set $\CN(\Hil)$ of all von Neumann subalgebras $\CM\subset B(\Hil)$ on a given Hilbert space $\Hil$. This is a bounded lattice with partial order given by inclusion and minimal/maximal elements $\Cl1$ and $B(\Hil)$, respectively. By the Bicommutant Theorem, the commutant map $\CM\mapsto\CM':=\{A\in B(\Hil)\col [A,M]=0\,\forall M\in\CM\}$ is an order-reversing involution, making $\CN(\Hil)$ an involution lattice. The resulting lattice operations are $\CM_1\wedge\CM_2=\CM_1\cap\CM_2$ and $\CM_1\vee\CM_2=(\CM_1\cup\CM_2)''$.

The commutant is not a complement because $\CM\cap\CM'=\Cl1$ only holds for factors. In particular, $\CN(\Hil)$ is not orthomodular (proper inclusions with trivial relative complements exist). $\CN(\Hil)$ carries a natural automorphic action of $\AU(\Hil)$. This lattice plays an important role in the operator-algebraic formulation of quantum field theory \cite{Haag:1996,Araki:1963}.

\section{Localization Observables}\label{section:LocalizationObservables}

\subsection{General observables on involution lattices}

In standard quantum mechanics, an observable is a selfadjoint operator on a Hilbert space $\Hil$, or, equivalently, a projection-valued measure $P:\Bor(\Rl^d)\to\CP(\Hil)$ (for $d>1$, this includes the case of several commuting selfadjoint operators and their joint spectral measure $P$, like in the case of the three-component position operator on $L^2(\Rl^3)$). Projection-valued measures are therefore usually regarded as the right formulation of (sharp) observables \cite{BuschLahtiMittelstaedt:1996}.

The defining properties of a projection-valued measure $P$, namely its normalization $P(\emptyset)=0$ and its $\sigma$-additivity, imply various other properties that are often seen as indispensible properties of a physical observable \cite{Wightman:1962}, and in fact they imply that $P:\Bor(\Rl^d)\to\CP(\Hil)$ is a logic homomorphism. For normalized $\sigma$-additive maps between more general logics, one has instead to demand more properties in the definition of an observable \cite[Ch.~III.2]{Varadarajan:2006}. We here give a further generalization of this concept that also covers the case of involution lattices. Following the usual terminology, we call such maps ``observables'' although it is clear from the generality of the definition below that an interpretation as a physical observable will require further assumptions.

\begin{definition}\label{def:Observable}
    Let $\CL$ and $\CQ$ be involution lattices.
    \begin{enumerate}
        \item An {\em observable} is a map $f:\CL\to\CQ$ that is
        \begin{enumerate}[{\textrm (i)}]
            \item\label{item:ObservableNormalization} {\em normalized}: $f(0)=0$ and $f(1)=1$,
            \item\label{item:ObservableAdditivity} {\em $\sigma$-additive} in the sense of \eqref{eq:additivity},
            \item\label{item:ObservableSeparation} \emph{separation preserving}: $A\leq B^{\perp}\Rightarrow f(A)\leq f(B)^\perp$ for all $A,B\in\CL$.
        \end{enumerate}

        \item A {\em finitely (completely) additive observable} is defined in the same manner, but with \ref{item:ObservableAdditivity} required only for finite (even for arbitrary) separated subsets of $\CL$.

        \item An observable $f:\CL\to\CQ$ is called {\em $G$-covariant} (with $G$ a group) if both $\CL$ and~$\CQ$ carry automorphic $G$-actions and $f$ is covariant in the sense of Definition~\ref{def:AutomorphicActionsAndCovariance}~\ref{item:Covariance}.
    \end{enumerate}
\end{definition}

For our purposes, $\CL$ will always be a logic (either $\Bor(\Rl^d)$ or $\CC(\Rl^{d+1})$) and $\CQ$ an involution lattice. We next derive some properties of such maps.

\begin{lemma}\label{lemma:Observables}
    Let $\CL$ be a logic, $\CQ$ an involution lattice, and $f:\CL\to\CQ$ an observable. Then, for any $A,B\in\CL$,
    \begin{enumerate}
        \item\label{item:fmonotone} $f$ is monotone: $A\leq B\Rightarrow f(A)\leq f(B)$,

        \item\label{item:fpreservesjoins} If $\CL$ is Boolean, $f$ preserves joins: $f(A\vee B)=f(A)\vee f(B)$.

        \item\label{item:fcomplement} $f(A^\perp)\leq f(A)^\perp$. If $\CQ$ is a logic, then $f(A^\perp)=f(A)^\perp$.

        \item\label{item:fbooleanlogiccase} If $\CL$ is Boolean and $\CQ$ a logic, then $f$ is a logic homomorphism.
    \end{enumerate}
\end{lemma}
\begin{proof}
    \ref{item:fmonotone} Let $A\leq B$. We have $B=A\vee (B\wedge A^\perp)$ by orthomodularity of $\CL$. As $A$ and $B\wedge A^\perp$ are separated, additivity of $f$ implies $f(B)=f(A)\vee f(B\wedge A^\perp)\geq f(A)$.

    \ref{item:fpreservesjoins} Using the distributive law in $\CL$, we decompose $A\vee B=A\vee(B\wedge A^\perp)$. As $A$ and $B\wedge A^\perp$ are separated and $f$ is monotone, we find
    \begin{align*}
        f(A)\vee f(B)
        &\leq
        f(A\vee B)=f(A)\vee f(B\wedge A^\perp)\leq f(A)\vee f(B),
    \end{align*}
    and the claim follows.

    \ref{item:fcomplement} The separation preservation of $f$ applied to the trivial inclusion $A^\perp\leq A^\perp$ results in the inequality $f(A^\perp)\leq f(A)^\perp$. If $\CQ$ is also orthomodular, application of the orthomodularity relation to $f(A^\perp)\leq f(A)^\perp$ gives
    \begin{align*}
        f(A)^\perp
        &=
        f(A^\perp)\vee\left[f(A)^\perp\wedge f(A^\perp)^\perp\right]
        =
        f(A^\perp)\vee\left[f(A)\vee f(A^\perp)\right]^\perp
        =
        f(A^\perp),
    \end{align*}
    where we have used that $f(A)\vee f(A^\perp)=1$ (which follows from normalization and separation preservation) in the last step.

    \ref{item:fbooleanlogiccase} This follows directly from the previous items and the de Morgan's laws.
\end{proof}

To illustrate the concept, we give some examples.

\begin{example}\label{example:Observables}
    \leavevmode
    {\em
    \begin{enumerate}
        \item Any logic homomorphism is an observable. In particular, the causal completion map $\cc$ in Lemma~\ref{lemma:CausalEmbedding} and the embeddings $\CP(\Hil)\hookrightarrow\CP_{\Rl}(\Hil)$ (Lemma~\ref{lemma:CPintoCPR}) and $\CP(\Hil)\hookrightarrow\CE(\Hil)$ (Section~\ref{section:SpectralOrder}) are observables that are Euclidean covariant and $\AU(\Hil)$-covariant, respectively.

        \item For $\CL$ a $\sigma$-algebra and $\CQ=\CP(\Hil)$, observables are the same as projection-valued measures. This example motivates the name ``observable''.

        \item\label{item:SystemOfImprimitivity} If $\CL$ is the $\sigma$-algebra of a standard Borel space $X$ on which a locally compact second countable group $G$ acts, and $G$ acts on $\CP(\Hil)$ as
        \begin{align}
            \alpha_g(P)=U(g)PU(g)^*,\qquad g\in G,
        \end{align}
        with a strongly continuous unitary representation $U$ of $G$, then a $G$-covariant observable $\CL\to\CP(\Hil)$ is the same as a \emph{system of imprimitivity} for $U$ based on $X$.

        \item A finitely additive observable between the logic of Minkowski spacetime $\CL=\CC(\Rl^{d+1})$ and the involution lattice $\CN(\Hil)$ of all von Neumann algebras on a Hilbert space $\Hil$ (Section~\ref{section:vonNeumannAlgebras}) is a quantum field theory respecting the usual locality, isotony and additivity properties of the Haag-Kastler framework \cite{Haag:1996}. If $\CC(\Rl^{d+1})$ is given its natural Poincaré action (see Example~\ref{example:MinkowskiLogic}), and the Poincaré group $\Poi(d+1)$ acts on $\CN(\Hil)$ as $\alpha_g(\CN)=U(g)\CN U(g)^*$ via a strongly continuous (anti)unitary representation $U$ of $\Poi(d+1)$, then the $\Poi(d+1)$-covariance of the observable expresses the Poincaré covariance of this quantum field theory.
    \end{enumerate}
    }
\end{example}

A general observable $f:\CL\to\CQ$ is not a lattice homomorphism, and in particular,~$f$ may fail to preserve joins for non-Boolean logics $\CL$, such as the causal logic $\CC(\Rl^{d+1})$ of Minkowski space. This failure of preserving joins is a property that should be expected of a (yet to be defined) spacetime localization observable: Motivated by the idea that the probability to find a quantum particle at an exact point $x\in\Rl^{d+1}$ should be zero, we anticipate that an observable $f:\CC(\Rl^{d+1})\to\CQ$ admitting such an interpretation should vanish on one point sets, $f(\{x\})=0$. But the join $\{x\}\vee\{y\}\in\CC(\Rl^{d+1})$ of two one point sets contains an open double cone in case $x$ lies timelike to~$y$, so it is reasonable to expect $f(\{x\}\vee\{y\})\neq0=f(\{x\})\vee f(\{y\})$.

\subsection{Covariant spatial and spacetime localization observables}

We will now sharpen the general concept of an observable as in Definition~\ref{def:Observable} in the following stricter Definitions~\ref{def:CovariantLocalizationObservables} and \ref{def:CausalLocalizationObservables} which require covariance and causality properties.

\begin{definition}[\bfseries covariant (spatial) localization observables]\label{def:CovariantLocalizationObservables}
    Let $\CQ$ be an involution lattice.
    \begin{enumerate}
        \item Let $G$ be a subgroup of the Euclidean group $\Eu(d)$. A {\em $G$-covariant spatial localization observable in $\CQ$} is a $G$-covariant observable $f:\Bor(\Rl^d)\to\CQ$ (an automorphic action of $G$ on $\CQ$ is assumed).

        \item\label{item:CovariantSpacetimeLocalizationObservable} Let $G$ be a subgroup of the Poincaré group $\Poi(d+1)$. A {\em $G$-covariant spacetime localization observable in $\CQ$} is a $G$-covariant observable $f:\CC(\Rl^{d+1})\to\CQ$ (an automorphic action of $G$ on $\CQ$ is assumed).
    \end{enumerate}
\end{definition}

As remarked in Example~\ref{example:Observables}~\ref{item:SystemOfImprimitivity}, covariant spatial localization observables taking values in $\CP(\Hil)$ with $G$-action given by a unitary representation are systems of imprimitivity. The most important example in this context are the Newton-Wigner operators \cite{NewtonWigner:1949}.

\begin{example}[\bfseries The Newton-Wigner operators]\label{example:NewtonWigner}
    {\em
        Let $U$ be a unitary representation of the Poincaré group $\Poi(d+1)$ on a Hilbert space $\Hil$, and let $P:\Bor(\Rl^d)\to\CP(\Hil)$ be an observable that is Euclidean covariant, where the Euclidean action on $\CP(\Hil)$ is given by $\alpha_g(Q)=U(g)QU(g)^*$ (in this situation, we will call the observable $U$-covariant for brevity), i.e. a system of imprimitivity for $\Eu(d)$. Then the corresponding {\em Newton-Wigner operators for $U$} are defined as
        \begin{align}
            \NW^U_k
            :=
            \int_{\Rl^d} x_k\,dP(\bx),\qquad k=1,\ldots,d.
        \end{align}
        If the representation $U$ is irreducible, Newton-Wigner operators exist if $U$ massive, or massless with helicity zero \cite[Cor.~9.17]{Varadarajan:2006}. By imposing invariance under the time reflection from $U$ and a regularity condition, these are also unique \cite{Wightman:1962}.

        For example, consider the scalar mass $m>0$ representation $U=U_{m,0}$. With $\om(\bp)=(m^2+\|\bp\|^2)^{1/2}$ the one-particle energy, we then have $\Hil=L^2(\Rl^d,\om(\bp)^{-1}d\bp)$ and
        \begin{align}\label{eq:Hilm}
            (U_{m,0}(x,\La)\psi)(\bp)
            &=
            e^{ip\cdot x}\psi([\La^{-1}p]),
        \end{align}
        where $p=(\om(\bp),\bp)$ and $[\La^{-1}p]\in\Rl^d$ denotes the spatial part of $\La^{-1}p\in\Rl^{d+1}$.

        The multiplication operator $\om^{-1/2}$ is a unitary map $\om^{-1/2}:\Hil\to L^2(\Rl^d,d\bp)$, and the Fourier transform is a unitary $\CF:L^2(\Rl^d,d\bx)\to L^2(\Rl^d,d\bp)$. Denoting the canonical projection-valued measure on $L^2(\Rl^d,d\bx)$ by $(P_0(A)f)(\bx)=\chi_A(\bx)f(\bx)$, the projections
        \begin{align}\label{eq:NWPVM}
            P^{\NW}_{U_{m,0}}:\Bor(\Rl^d)
            \to\CP(\Hil),\qquad
            A\mapsto
            \om^{1/2}\CF P_0(A)\CF^{-1}\om^{-1/2},
        \end{align}
        form the regular time-reflection invariant system of imprimitivity. We refer to \cite{Moretti:2023} for a recent discussion of various properties of these measures and their associated Newton-Wigner operators.
        }
\end{example}

\subsection{Translationally covariant causal localization observables}

We now consider causality properties of localization observables that are covariant under the spatial or spacetime translation group. For the time being, Lorentz invariance will play no role and we will effectively choose a reference frame given by a fixed time zero Cauchy surface $\{0\}\times\Rl^d\subset\Rl^{d+1}$ of Minkowski spacetime. The time vector $(1,\mathbf{0})$ in this reference frame will be denoted $e_0$.

The most elementary notion of causality, suitable for purely spatial formulations, is to require a dynamics with finite speed of propagation. As we are interested in massive particles moving at velocities smaller than the speed of light, we will work with the open forward lightcone $V^+$, and write
\begin{align}
    \bB:=\{\bx\in\Rl^d\col\|\bx\|<1\}
\end{align}
for the open unit ball in $\Rl^d$.

\begin{definition}[\bfseries causal spatial localization observables]\label{def:CausalLocalizationObservables}
    Let $\CQ$ be an involution lattice with automorphic $\Rl^d$-action $\sigma$, and $f:\Bor(\Rl^d)\to\CQ$ an $\Rl^d$-covariant spatial localization observable.
    \begin{enumerate}
        \item A {\em dynamics} is a one-parameter group $(\tau_t)_{t\in\Rl}\in\Aut\CQ$ that commutes with the $\Rl^d$-action $\sigma_{\bx}\in\Aut\CQ$, $\bx\in\Rl^d$, and thus extends it to an $\Rl^{d+1}$-action $\alpha_{(t,\bx)}=\tau_t\sigma_{\bx}$.

        \item A dynamics $\tau$ has {\em speed of propagation bounded by} $0<v<\infty$ if the {\em speed limit condition}
        \begin{align}\label{eq:BoundedSpeed}
            \tau_t(f(A)) \leq
            f([A+|t|v\bB]^\perp)^\perp,\qquad A\in\Bor(\Rl^d),
        \end{align}
        holds. The dynamics satisfies the \emph{strong speed limit condition} if even
        \begin{align}\label{eq:BoundedSpeedStrong}
            \tau_t(f(A)) \leq
            f(A+|t|v\bB),\qquad A\in\Bor(\Rl^d)
        \end{align}
        holds.
    \end{enumerate}
\end{definition}

Note that the speed limit condition \eqref{eq:BoundedSpeed} is weaker than the strong speed limit condition \eqref{eq:BoundedSpeedStrong} in a general involution lattice, but equivalent to it if $\CQ$ is a logic (Lemma~\ref{lemma:Observables}~\ref{item:fcomplement}).

Consider two Borel sets $A,B$ that are a distance $d(A,B)>0$ apart. Then $A+|t|v\bB\subset B^\perp$ for $|t|v\leq d(A,B)$, and the speed limit condition implies that $\tau_t(f(A))$ and $f(B)$ remain separated in $\CQ$ for such times.

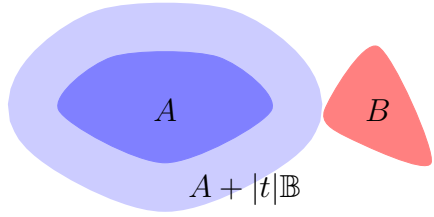
\begin{figure}[h]
    \centering
    \begin{tikzpicture}[scale=0.7]

        \filldraw[draw=blue!20, fill=blue!20, line width=13mm] plot [smooth cycle] coordinates {(0,0) (1,1) (3,1) (4,0) (2,-1)};

         \filldraw[draw=blue!50, fill=blue!50] plot [smooth cycle] coordinates {(0,0) (1,1) (3,1) (4,0) (2,-1)};

        \filldraw[draw=red!50, fill=red!50] plot [smooth cycle] coordinates {(5,0) (6,1.2) (7,-1) (5.4,-0.5)};

        \node at (2,0) {$A$};
        \node at (3.5,-1.5) {$A+|t|\bB$};
        \node at (6,0) {$B$};
    \end{tikzpicture}

    \caption{As the spatial regions $A+|t|\bB$ and $B$ remain separated for small times, $\tau_t(f(A))$ remains separated from $f(B)$ for a causal spatial localization observable $f$.}
\end{figure}

In case $\CQ$ is a logic, the speed limit condition equals the strong speed limit condition and becomes very similar to (a consequence of) Castrigiano's causality condition for projection-valued measures \cite{Castrigiano:2017}. To explain the idea, think of the elements $f(A)$ as representing the probability that a particle is localized in the region $A$. With $\tau_t$ representing the time evolution, a particle that is localized in $A$ at time $t=0$ and moves with a velocity less than the speed limit $v$, will at time $t>0$ be localized in the enlarged region $A+tv\bB$. As also particles that were localized outside of $A$ at $t=0$ may enter $A+tv\bB$ at $t>0$, we expect $ \tau_t(f(A))\leq f(A+|t|v\bB)$. Conversely, if a particle is localized in $A+|t|v\bB$ for $t<0$, it may, but does not have to, enter $A$ at $t=0$, leading again to the inequality $ \tau_t(f(A))\leq f(A+|t|v\bB)$.

\medskip

We also want to consider {\em spacetime} localization observables $F:\CC(\Rl^{d+1})\to\CQ$, which by definition satisfy the causality condition $\CO_1\subset\CO_2'\Rightarrow F(\CO_1)\leq F(\CO_2)^\perp$ (Definition~\ref{def:CovariantLocalizationObservables}~\ref{item:CovariantSpacetimeLocalizationObservable}). In order to connect this causal behaviour to {\em spatial} localization observables with finite speed of propagation, we first recall the definition of the {\em chronological future/past} $I^\pm$ and $I$ \eqref{eq:defI}, and agree to write $I^{(\pm)}(A):=I^{\pm}(\{0\}\times A)$ for $A\subset\Rl^{d}$.  Denoting by $\Sigma=\{0\}\times\Rl^d$ the time zero Cauchy surface, one easily checks
\begin{align}\label{eq:Ivstfat}
    (I(A)+x)\cap\Sigma
    =
    A+\bx+|x_0|\bB
    ,
    \qquad
    A\in\Bor(\Rl^d),\quad x\in\Rl^{d+1}.
\end{align}
In the following lemma, we will be interested in Borel sets $A,B\in\Bor(\Rl^d)$ and spacetime vectors $x,y\in\Rl^{d+1}$ such that $\cc(A)+x\subset(\cc(B)+y)'$. Taking into account the translation covariance of $\cc$ and the causal complement, as well as $I(\CO)'=\Rl^{d+1}\setminus I(\CO)$, and \eqref{eq:Ivstfat}, it becomes apparent that
\begin{align}\label{eq:geometry}
    \cc(A)+x\subset(\cc(B)+y)'
    \Longleftrightarrow
    B\subset [A+(\bx-\by)+|x_0-y_0|\bB]^\perp.
\end{align}

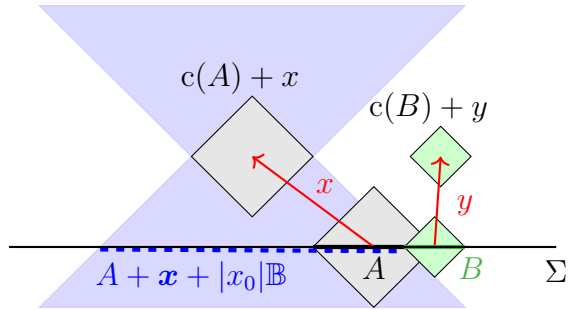
\begin{figure}[h]
    \centering
    \begin{tikzpicture}[scale=0.8]
        \draw[draw=blue!20,fill=blue!15!white] (-4,1.5) -- ++(-2.5,2.5) --  ++(7,0) -- (-2,1.5) -- ++(2.5,-2.5) -- ++(-7,0) -- cycle;

        \draw[draw=black,fill=black!10!white] (-2,0) -- (-1,1) -- (0,0) -- (-1,-1) -- cycle;
        \draw[draw=black,fill=black!10!white] (-4,1.5) -- (-3,2.5) -- (-2,1.5) -- (-3,0.5) -- cycle;

        \draw[ultra thick] (-2,0) to (0,0);
        \draw[->,red,thick] (-1,0) to (-3,1.5);

        \node at (-1.8,1) {\textcolor{red}{$x$}};
        \node at (-3.2,2.8) {$\cc(A)+x$};

        \definecolor{clr2}{RGB}{31,182,83};
        \draw[draw=black,fill=green!20!white] (-0.5,0) -- (0,0.5) -- (0.5,0) -- (0,-0.5) -- cycle;

        \begin{scope}[shift={(0.1,1.5)}]
            \draw[draw=black,fill=green!20!white] (-0.5,0) -- (0,0.5) -- (0.5,0) -- (0,-0.5) -- cycle;
        \end{scope}
        \draw[->,red,thick] (0,0) to ++(0.1,1.5);
        \node at (0.5,0.7) {\textcolor{red}{$y$}};
        \node at (-0.1,2.3) {$\cc(B)+y$};

        \draw[ultra thick,draw=DGr] (-.5,-0) to (0.5,-0);

        \node at (-1,-0.35) {$A$};
        \node at (0.6,-0.35) {\textcolor{DGr}{$B$}};

        \draw[ultra thick,dashed,draw=blue] (-5.5,-0.04) to (-0.5,-0.07);
        \node[blue] at (-4,-0.5) {$A+\bx+|x_0|\bB$};

        \draw[thick,draw=black] (-7,0) to (2,0);
        \node[black] at (2,-0.4) {$\Sigma$};
    \end{tikzpicture}
\caption{Illustration of the configurations in \eqref{eq:Ivstfat} and \eqref{eq:geometry}.}
\end{figure}

\begin{lemma}
    An $\Rl^d$-covariant spatial localization observable $f:\Bor(\Rl^d)\to\CQ$ with dynamics $\tau$ and induced $\Rl^{d+1}$-action $\alpha$ has speed of propagation bounded by $1$ if and only if it is {\em causal} in the sense that for all $A,B\in\Bor(\Rl^d)$ and all $x,y\in\Rl^{d+1}$, the implication
    \begin{align}\label{eq:LocalDynamics}
        \cc(A)+x \subset (\cc(B)+y)'
        \Rightarrow
        \alpha_x(f(A)) \leq \alpha_y(f(B))^\perp
    \end{align}
    holds.
\end{lemma}
\begin{proof}
    By translation covariance, it suffices to consider $y=0$. By \eqref{eq:geometry}, the assumption is then equivalent to $B\subset[A+\bx+|x_0|\bB]^\perp$. If $f$ has speed of propagation bounded by~$1$, the speed limit condition \eqref{eq:BoundedSpeed} yields
    \begin{align}
        \alpha_x(f(A))
        =
        \tau_{x_0}(f(A+\bx))
        \leq
        f([A+\bx+|x_0|\bB]^\perp)^\perp
        \leq
        f(B)^\perp
        .
    \end{align}
    Thus $f$ is causal in the sense of \eqref{eq:LocalDynamics}.

    Conversely, if $f$ is causal, let $t\in\Rl$, $A\in\Bor(\Rl^d)$ and set $B:=(A+|t|\bB)^\perp$. According to \eqref{eq:geometry}, we then have $\cc(A)+te_0\subset\cc(B)'$, and \eqref{eq:LocalDynamics} yields
    \begin{align}
        \tau_t(f(A))
        \leq
        f(B)^\perp
        =
        f([A+|t|\bB]^\perp)^\perp,
    \end{align}
    which shows that\footnote{ In this last step, we can not conclude the sharper inequality $\tau_t(f(A)) \leq f(A+|t|\bB)$, which is why we base our analysis on \eqref{eq:BoundedSpeed}.} $f$ has propagation speed bounded by $1$.
\end{proof}

Thanks to this lemma, we call spatial localization observables with speed of propagation bounded by $1$ \emph{causal}. We can now can easily pass from spatial to spacetime localization observables and vice versa.

\begin{proposition}\label{prop:SpacetimeObservableInduceCausalSpatialObservables}
    Let $\CQ$ be a complete involution lattice.
    \begin{enumerate}
        \item\label{item:SpacetimeToSpatial} Let $F:\CC(\Rl^{d+1})\to\CQ$ be a translationally covariant spacetime localization observable (w.r.t. some automorphic $\Rl^{d+1}$-action $\alpha$ on $\CQ$). Then
        \begin{align}
            f := F\circ \cc : \Bor(\Rl^d) \to \CQ
        \end{align}
        is a translationally covariant spatial localization observable that is causal w.r.t. the dynamics $\tau_t:=\alpha_{(t,\boldsymbol{0})}$ (with $\cc$ defined in Definition~\ref{def:c}).

        \item\label{item:SpatialToSpacetime} Let $f:\Bor(\Rl^d)\to\CQ$ be an observable with causal dynamics $\tau$, and let $\CQ$ be complete. Then
        \begin{align}\label{eq:F}
            F:\CC(\Rl^{d+1})\to\CQ,
            \qquad
            F(\CO)
            :=
            \bigvee_{x\in\Rl^{d+1},A\in\Bor(\Rl^d)\atop \cc(A)+x\subset\CO} \alpha_x(f(A))
        \end{align}
        is a spacetime localization observable.
    \end{enumerate}
\end{proposition}
\begin{proof}
    \ref{item:SpacetimeToSpatial} As $f$ is the composition of an observable with a $\sigma$-additive homomorphism (Lemma~\ref{lemma:CausalEmbedding}), it is also an observable. Clearly $f$ is translationally invariant and $\tau_t:=\alpha_{(t,\boldsymbol{0})}$ is a dynamics. For every $t\in\Rl$ and every $A\in\Bor(\Rl^d)$, we have the geometric inclusion $\cc(A)+te_0 \subset \cc(A+|t|\bB)$. This implies
    \begin{align}
        \tau_{t}(f(A))
        =
        F(\cc(A)+te_0)
        \leq
        F(\cc(A+|t|\bB))
        =
        f(A+|t|\bB)
        \leq
        f([A+|t|\bB]^\perp)^\perp
        ,
    \end{align}
    i.e. $f$ is causal.

    \ref{item:SpatialToSpacetime} As $\CQ$ is complete, $F$ is well defined. It is clear from its definition that $F$ is normalized. In case $(\CO_n)_{n\in\Nl}\subset\CC(\Rl^{d+1})$ is separated, then in particular $\CO_n\cap\CO_m=\emptyset$ for $n\neq m$. Hence any region $\cc(A)+x$ (with $A\in\Bor(\Rl^d)$ and $x\in\Rl^{d+1}$) can lie in at most one of the $\CO_n$, and then \eqref{eq:F} implies the additivity of $F$.

    For checking that $F$ preserves separation, we use that $f$ is causal: For any $A,B\in\Bor(\Rl^d)$ and $x,y\in\Rl^{d+1}$ such that $\cc(A)+x\subset(\cc(B)+y)'$, we have $\alpha_x(f(A))\leq\alpha_y(f(B))^\perp$. In view of the definition of $F$, this implies $F(\CO_1)\leq F(\CO_2)^\perp$ for $\CO_1,\CO_2\in\CC(\Rl^{d+1})$ with $\CO_1\subset\CO_2'$.
\end{proof}

In quantum field theory, this simple idea of promoting a spatial localization observable to a spacetime localization observable is connected to time zero fields \cite{Schlingemann:1999-1,GrosseLechnerLudwigVerch:2013}. In our present formulation, it is entirely based on Cauchy surfaces that are parallel translates of the time zero surface $\Sigma=\{0\}\times\Rl^d$, and does not make any explicit reference to Lorentz symmetry. We refer to \cite{Castrigiano:2017,Moretti:2023,DeRosaMoretti:2024} for works that consider Cauchy surfaces that are arbitrary Poincaré transforms of $\Sigma$, or even general curved Cauchy surfaces.

\subsection{No localization observables in $\CP(\Hil)$}

Now that causal spatial and spacetime observables are defined, we can investigate their existence and uniqueness properties. These will of course depend on the target lattice $\CQ$, including its translational symmetries and dynamics. We will work in the usual setting of quantum theory and consider involution lattices contained in operators on a Hilbert space~$\Hil$, and strongly continuous unitary representations $U$ of $\Rl^{d+1}$ acting on these operators by $\alpha_x(P)=U(x)PU(x)^*$. We will furthermore restrict to the physically interesting case that $U$ has {\em positive energy}, namely that the energy-momentum spectrum of $U$ (defined to be the joint spectrum of the commuting selfadjoint generators of this representation) is contained in the closed forward light cone $\overline{V^+}$.

We begin with a No-Go theorem concerning the well-studied case $\CQ=\CP(\Hil)$. This should be seen in line with many well established No-Go results which go back to an article by Schlieder \cite{Schlieder:1971}. The version we present here is essentially identical to the one given by Jancewicz \cite{Jancewicz:1977}, and a somewhat stronger version has been given by Malament \cite{Malament:1996}.

\begin{theorem}\label{thm:NoGo}
   Let $U$ be a positive energy representation $\Rl^{d+1}$ on some Hilbert space $\Hil$ and $P:\Bor(\Rl^d)\to \CP(\Hil)$ a map with all properties of a translationally covariant causal spatial localization observable except $\sigma$-additivity (Definition~\ref{def:CausalLocalizationObservables}), where it is understood that the spatial translation symmetry and the dynamics are given by $U$. Then $P(A)=0$ for all bounded $A$.

   In particular, there exists no translationally covariant causal spatial localization observable $P:\Bor(\Rl^d)\to\CP(\Hil)$, and no translationally covariant spacetime localization observable $P:\CC(\Rl^{d+1})\to\CP(\Hil)$.
\end{theorem}

The proofs of these theorems are usually based on a non-trivial analyticity argument due to Borchers \cite{Borchers:1967}. Depending on the precise nature of the causality assumption\footnote{In both the arguments given by Jancewicz and Malament it is assumed that the orthogonal projections $P,Q$ corresponding to two disjoint spatial regions are orthogonal to each other. For time translates by sufficiently small parameter $|t|$, Jancewicz assumes $PU(t)QU(-t)=0$ (as also done here), whereas Malament only assumes $[P,U(t)QU(-t)]=0$. The conclusions are the same in both cases, but only under the weaker (commutator) assumption, one needs to invoke Borchers' theorem.}, the use of Borchers' theorem is however not necessary. For the sake of self-containedness, we state the simplified analyticity statement as a lemma and then provide a very short proof of Theorem~\ref{thm:NoGo}.

\begin{lemma}\label{borcherseinfach}
    Let $\Hil$ be a Hilbert space, $U$ a positive energy representation of $\Rl^{d+1}$ on $\Hil$, and $P,Q\in\CP(\Hil)$ such that $U(x)PU(x)^{-1}\leq Q$ for all $x$ in a neighbourhood of $0\in\Rl^{d+1}$. Then $U(x)PU(x)^{-1}\leq Q$ for all $x\in\Rl^{d+1}$.
\end{lemma}
\begin{proof}
    Let $\xi,\psi\in\Hil$ and consider the function $\varphi:\Rl^{d+1}\to\Cl$, $\varphi(x):=\langle Q^\perp\xi,U(x)P\psi\rangle$. Our assumptions say that $\varphi(x)=0$ for all $x$ in a neighbourhood of the origin. But as a consequence of the positive energy assumption on the representation, $\varphi$ has an analytic continuation to the forward tube $\Rl^{d+1}+iV^+$ which is continuous on the closure of this domain. This implies that $\varphi$ vanishes identically and in particular on $\Rl^{d+1}$. As the vectors $\xi$ and $\psi$ were arbitrary, the claim follows.
\end{proof}

\begin{proof}(of Theorem~\ref{thm:NoGo})
    Let $A\in\Bor(\Rl^d)$ be bounded and pick $\by\in\Rl^d$ such that $A$ and $A+\by$ have non-zero spatial distance $d(A,A+\by)>0$. Let $P:\Bor(\Rl^d)\to\CP(\Hil)$ be a map as described in the statement of the theorem. Then, by causality of $P$
    \begin{align*}
        U(x)P(A+\by)U(x)^*
        &\leq
        P((A+\by+\bx+|x_0|\bB)^\perp)^\perp
        \leq
        P(A)^\perp,
    \end{align*}
    where the last estimate holds whenever $A+\by+\bx+|x_0|\bB\subset A^\perp$, which is the case whenever $\|\bx\|+|x_0|<d(A,A+\by)$. Hence the preceding lemma implies $U(x)P(A+\by)U(x)^*\leq P(A)^\perp$ for all $x\in\Rl^{d+1}$. For $x=(0,-\by)$, this results in $P(A)\leq P(A)^\perp$ and $P(A)=P(A)\wedge P(A)^\perp=0$ because $\CP(\Hil)$ is orthocomplemented.

    Now assume that $P$ also satisfies the additivity assumption in Definition~\ref{def:Observable}~\ref{item:ObservableAdditivity}, i.e. is a translationally covariant causal spatial localization observable $P:\Bor(\Rl^d)\to\CP(\Hil)$. Pick a covering $(A_n)_{n\in\Nl}$ of $\Rl^d$ by bounded Borel sets. Then $0=\bigvee_n P(A_n)=P(\bigvee_n A_n)=P(\Rl^d)=1$, a contradiction. Hence such $P$ does not exist. Since Proposition~\ref{prop:SpacetimeObservableInduceCausalSpatialObservables} shows that any spacetime localization observable induces a causal spatial localization observable, also the last statement follows.
\end{proof}

As is well known, this result implies in particular that the Newton-Wigner observable (Example~\ref{example:NewtonWigner}) with the dynamics coming from its defining Poincaré representation is not causal. It is therefore generally accepted that the framework of $\CP(\Hil)$ is too narrow to allow for satisfactory relativistic localization observables.

We will therefore extend our framework. In doing so, we will follow a different route than the one passing through POVMs and effects mentioned in the Introduction. The main idea for our extension is to consider the larger involution lattice $\CP_{\Rl}(\Hil)\supset\CP(\Hil)$ instead of $\CP(\Hil)$, with spatial translations and dynamics given by a positive energy representation $U$ of $\Rl^{d+1}$ (i.e., $\alpha_x(E)=U(x)EU(x)^{-1}$).

In this lattice, the situation considered in Lemma~\ref{borcherseinfach} now amounts to two real orthogonal projections $E,F\in\CP_{\Rl}(\Hil)$ such that $E\leq U(x)F'U(x)^*$ for all $x$ in a neighbourhood of zero. As the involution $\CP_{\Rl}(\Hil)\ni E\mapsto E'$ replacing the orthogonal complement is given by the symplectic complement of the range of $E$, this implies only that the imaginary part $\imag\varphi(x)$ of $\varphi(x):=\langle E\xi,U(x)F\psi\rangle$ vanishes for sufficiently small $x$. This is clearly a weaker condition on the analytic function $\varphi:\Rl^{d+1}+iV^+\to\Cl$ than in the complex linear case. In particular, it does not imply that $\imag\varphi$ has to vanish identically, i.e. $E\leq U(x)F'U(x)^*$ for all $x\in\Rl^d$ does {\em not} follow.

To give an explicit example, consider the irreducible positive energy Poincaré representation $U$ of mass $m>0$ and spin zero on the Hilbert space \eqref{eq:Hilm}, with translations
\begin{align}
    (U(x)\psi)(\bp)=e^{i(x_0\om_{\bp}-\bx\bp)}\psi(\bp),\qquad \bp\in\Rl^d.
\end{align}
As vectors in the ranges of two real projections $E_1,E_2\in\CP_{\Rl}(\Hil)$, we take the Fourier transforms $\tilde f_k$ of two real functions $f_k\in C_c^\infty(\Rl^{d})$, $k=1,2$, with disjoint supports. Then
\begin{align}
    \imag\langle \tilde f_1,U(x)\tilde f_2\rangle
    &=
    \imag\int\frac{d\bp}{\om_{\bp}}\overline{\tilde f_1(\bp)}e^{i(x_0\om_{\bp}-\bx\bp)}\tilde f_2(\bp)
    \\
    &=
    (2\pi)^{-d}
    \int d\by\,d\bz\,f_1(\by)f_2(\bz)
        \int\frac{d\bp}{\om_{\bp}} \sin(x_0\om_{\bp}+i\bp\cdot(\by-\bz-\bx)).
\end{align}
The distribution $x\mapsto\int\frac{d\bp}{\om_{\bp}}\ \sin(x_0\om_{\bp}+i\bp\cdot\bx)$ is the commutator function of the Klein Gordon quantum field. As it vanishes on all spacelike points, and $\|\by-\bz\|\geq d(\supp(f_1),\supp(f_2))>0$, we see $ \imag\langle\tilde f_1,U(x)\tilde f_2\rangle=0$ for all sufficiently small $x$. However, since the commutator distribution is not zero, there exist $f_1,f_2$ as described above such that $\imag\langle\tilde f_1,U(x)\tilde f_2\rangle\neq0$ for large $x$. This indicates that the drastic implications of causal behavior that we observed for $\CP(\Hil)$ do not occur for $\CP_{\Rl}(\Hil)$.

\subsection{Constraints on causal localization observables in $\CP_{\Rl}(\Hil)$}\label{subsection:CausalLocalizationObservablesPRH}

We now consider a positive energy representation $U$ of $\Rl^{d+1}$, endowing $\CP_{\Rl}(\Hil)$ with $\Rl^d$-translation symmetry and dynamics, and ask how Theorem~\ref{thm:NoGo} transfers to this situation.

We begin with some remarks comparing $\CP_{\Rl}(\Hil)$ and $\CP(\Hil)$. We had already observed that $\CP(\Hil)\subset\CP_{\Rl}(\Hil)$ is a sublogic (Proposition~\ref{lemma:CPintoCPR}), i.e. a subset such that the lattice operations and involution of $\CP_{\Rl}(\Hil)$ restrict to the lattice operations and orthocomplementation of $\CP(\Hil)$. Hence a study of the localization observables mapping into $\CP_{\Rl}(\Hil)$ strictly contains a study of those mapping into $\CP(\Hil)$. However, we will see that most of the subspaces appearing in the range of such a map are maximally non-complex in a certain sense, and introduce some terminology for that now.

Any real linear closed subspace $H\subset\Hil$ naturally defines two closed complex linear subspaces: The largest closed complex subspace contained in $H$, denoted $H_{\Cl}$, and the smallest closed complex subspace containing $H$, denoted $H^{\Cl}$. These are given by
\begin{align}
    H_{\Cl}=H\cap iH \subset H \subset \overline{H+iH}=H^{\Cl}
\end{align}
and satisfy
\begin{align}
    (H')^{\Cl}
    &=
    \overline{H'+iH'}
    =
    (H\cap iH)'
    =
    (H_{\Cl})'=(H_{\Cl})^{\perp_{\Cl}}
    ,
    \\
    (H')_{\Cl}&=(H^{\Cl})'=(H^{\Cl})^{\perp_{\Cl}}
    =
    H^{\perp_{\Cl}}.
\end{align}
Clearly $H$ is complex linear if and only if $H_{\Cl}=H=H^{\Cl}$. A few relevant definitions and facts in this context are the following.

\begin{definition}
    Let $H\in\CP_{\Rl}(\Hil)$.
    \begin{enumerate}
        \item $H$ is called {\em separating} if $H_{\Cl}$ is trivially small, namely $H_{\Cl}=\{0\}$ (meaning that $H$ is ``maximally non-complex''),
        \item $H$ is called {\em cyclic} if $H^{\Cl}$ is trivially large, namely $H^{\Cl}=\Hil$, meaning that its complex span is dense in $\Hil$,
        \item $H$ is called {\em standard} if it is cyclic and separating.
    \end{enumerate}
\end{definition}

We will use the elementary facts (see \cite{Longo:2008_2,CorreadaSilvaLechnerLongo:2026}) that $H$ is cyclic if and only if its symplectic complement $H'$ is separating, and that in an inclusion $K\subset H$ cyclicity of $K$ implies cyclicity of $H$, while $H$ being separating implies $K$ being separating. This terminology and notation will also be used for the associated projections.

We now show that these properties naturally arise in the context of causal localization observables by a Reeh-Schlieder type argument \cite{ReehSchlieder:1961}.

\begin{proposition}\label{prop:standardsubspacesarenecessary}
    Let $U$ be a positive energy representation of $\Rl^{d+1}$ on some Hilbert space $\Hil$, and $E:\Bor(\Rl^d)\to(\CP_{\Rl}(\Hil),U)$ a translationally covariant causal spatial localization observable. Then $E(A)$ is cyclic for every $A\in\Bor(\Rl^d)$ with non-empty interior, and separating for every $A\in\Bor(\Rl^d)$ such that $A^\perp\in\Bor(\Rl^d)$ has non-empty interior.
\end{proposition}
\begin{proof}
    The idea of the proof is to consider the two maps
    \begin{align}
        E_{\Cl},E^{\Cl}:\Bor(\Rl^d)\to\CP(\Hil),\qquad E_{\Cl}(A):=E(A)_{\Cl},\quad E^{\Cl}(A):=E(A)^{\Cl},
    \end{align}
    mapping into complex linear projections onto the ``small'' and ``large'' complex subspaces generated by $E(A)$, respectively.

    Let $A\in\Bor(\Rl^d)$ have non-empty interior and let $P\in\CP(\Hil)$ satisfy $P\leq E(A)^{\perp_{\Cl}}=E^{\Cl}(A)^{\perp_{\Cl}}$. Then there exists open $A_0\subset A$ and $\eps>0$ such that $A_0+\bx\subset A$ for $\bx\in\Rl^d$ with $\|\bx\|<\eps$. Lemma~\ref{borcherseinfach} implies $P\leq E(A_0+x)^{\perp_{\Cl}}$ for all $x\in\Rl^{d+1}$. Picking vectors $(x_n)\in\Rl^{d+1}$ such that $\bigvee_n(A_0+x_n)=\Rl^{d+1}$, this gives
    \begin{align}
        P\leq\bigwedge_n E(A_0+x_n)^{\perp_{\Cl}}
        =
        \left(\bigvee_n E(A_0+x_n)\right)^{\perp_{\Cl}}
        =
        1^{\perp_{\Cl}}=0,
    \end{align}
    namely $P=0$. This implies $E(A)^{\Cl}=1$, i.e. $E(A)$ is cyclic.

    Since $E(A)\leq E(A^\perp)'$, we also see that in case $A^\perp$ has non-empty interior, then $E(A^\perp)'$ is cyclic, so $E(A^\perp)$ is separating, which implies that its subprojection $E(A)$ is separating as well.

    It is interesting to note that in our present context, one can also give a direct argument for the separating property in case $A$ is bounded: We claim that $E_{\Cl}$ has all properties of a translationally covariant causal spatial localization observable except additivity (Definition~\ref{def:Observable}~\ref{item:ObservableAdditivity}). Clearly
    \begin{align}
        E_{\Cl}(\emptyset)=0, \qquad E_{\Cl}(\Rl^d)=1,
    \end{align}
    and for $A,B\in\Bor(\Rl^d)$
    \begin{align}
        A\subset B^\perp \Rightarrow
        E_{\Cl}(A) \leq E(A) \leq E(B)' \leq E_{\Cl}(B)' = E_{\Cl}(B)^{\perp_{\Cl}}.
    \end{align}
    The covariance of $E$ passes down to $E_{\Cl}$ because the representation $U$ is complex linear. Similarly one checks that $E_{\Cl}$ also satisfies the speed limit condition in Definition~\ref{def:CausalLocalizationObservables} from $E$. We may therefore apply Theorem~\ref{thm:NoGo} and conclude $E_{\Cl}(A)=0$, which is equivalent to $E(A)$ being separating, for all bounded $A\in\Bor(\Rl^d)$.
\end{proof}

Proposition~\ref{prop:standardsubspacesarenecessary} shows that the spaces $H(A):=E(A)\Hil$ are standard if $A\in\Bor(\Rl^d)$ and~$A^\perp$ have interior points. Standard subspaces might appear to be exotic at first sight, but exist in abundance (very simple examples are given by the closed real span of an orthonormal basis of $\Hil$). We therefore conclude that the constraints on causal spatial localization observables in $\CP_{\Rl}(\Hil)$ are non-trivial but weaker than in the complex linear case.

Standard subspaces lie at the basis of modular theory. A crucial aspect of a standard subspace~$H$ is that it defines a unitary one-parameter group $\Delta_H^{it}$, $t\in\Rl$, its modular group, and an antiunitary involution $J_H$ commuting with it, in such a way that $H$ can be recovered from these data (see \cite{Longo:2008_2, CorreadaSilvaLechner:2025}, and the upcoming textbook \cite{CorreadaSilvaLechnerLongo:2026}, and Appendix~\ref{appendix:StandardSubspaces} for a minimal review of this structure).

\subsection{Lorentz symmetry and wedges} In addition to the constraints demanding cyclic/separating properties of the local subspaces $H(A)$, there exist further constraints on causal localization observables $E:\Bor(\Rl^d)\to\CP_{\Rl}(\Hil)$ that we will explain now. Although all symmetry assumptions made so far refer only to translations, the causal structure of $\CC(\Rl^{d+1})$ is implicitly linked to Poincaré symmetry, and it turns out that this results in a Lorentz invariance constraint.

To derive this constraint, we consider half spaces at time zero,
\begin{align}\label{eq:TimeZeroHalfSpace}
    A(\by)
    :=
    \{\bx\in\Rl^d \col \bx\by\geq0\},
\end{align}
where $\by\in\Rl^d$ is any unit vector (the inward pointing normal of the boundary of $A(\by)$) and $\bx\by$ denotes the Euclidean scalar product. The causal completion of this half space is
\begin{align}
    W(\by)
    :=
    \cc(A(\by))
    =
    \{x\in\Rl^{d+1} \col |x_0|\leq \bx\by\},
\end{align}
the (closed) Rindler wedge in direction of $\by$. We note that it has the invariance properties
\begin{align}
    \La_\by(t)W(\by)&=W(\by),\qquad t\in\Rl,
    \\
    W(\by)+sy_\pm &\subset W(\by),\qquad s\geq0,\qquad y_\pm:=(\pm1,\by)
\end{align}
where $\La_\by(t)$ is the Lorentz boost in $\by$-direction with rapidity $t$. The lightlike vectors $y_\pm$ in the boundary of $W(\by)$ satisfy (cf. the speed limit condition in \eqref{eq:BoundedSpeed} and \eqref{eq:geometry})
\begin{align}\label{eq:AyInvariance2}
    A(\by)+s\by+|s|\bB \subset A(\by),\qquad s\geq0.
\end{align}

\begin{figure}[h]
    \centering
    \begin{tikzpicture}[scale=0.67]
    \coordinate (center) at (0,0);

        \begin{scope}[rotate around={-90:(center)}]
            \draw[draw=green!50!white,fill=green!5!white,thick] (0.2,3.2) -- (3.4,0) -- (6.6,3.2); 
            \draw[draw=blue!50!white,fill=blue!5!white,thick] (0.2,-3.2) -- (3.4,0) -- (6.6,-3.2); 

            \node at (3.0,1.0) {$\bullet$};
            \node at (3.8,-1.0) {$\bullet$};

            \draw[black,thick,dashed,->] (3.0,1.0) -- (3.8,-1.0);

        \begin{axis}[
        axis line style={draw=none},
        ytick style={draw=none},
        xtick style={draw=none},
        yticklabel={\empty},
        xticklabel={\empty},
        domain=-3:3, 
        samples=100, 
        ymin=0, 
        ymax=6, 
    ]
        \addplot[black, mark=none] {sqrt(x^2+0.15)};
        \addplot[black, mark=none] {sqrt(x^2+0.5)};
        \addplot[black, mark=none] {sqrt(x^2+1)};
        \addplot[black, mark=none] {sqrt(x^2+2)};
        \addplot[black, mark=none] {sqrt(x^2+4)};
        \addplot[black, mark=none] {sqrt(x^2+7)};
    \end{axis}
    \draw[draw=white,fill=white] (0.2,3.2) -- (6.6,3.2) -- (6.6,3.8) -- (0.2,3.8) -- cycle;
    \end{scope}

    \node at (0.95,-5.2) {$W$};
    \node at (1.3,-2.7) {$x$};
    \node at (-1.1,-3.3) {$j_{e_1}(x)$};
    \end{tikzpicture}
    \caption{The right wedge $W=W(e_1)$, the orbits of its boost group $\La_{e_1}$, and the reflection $j_{e_1}$.}
\end{figure}
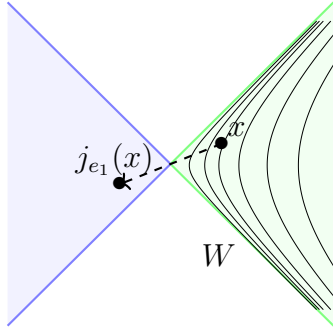

As one more geometric datum entering our discussion, we define $j_\by:\Rl^{d+1}\to\Rl^{d+1}$ to be the reflection at the edge of $W(\by)$, i.e.
\begin{align}
    j_\by(\by):=-\by,\quad j_\by(e_0):=-e_0,
    \qquad j_\by(\bx):=\bx,\quad \bx\by=0.
\end{align}
The formulation and proof of the following proposition relies on the strong speed limit condition \eqref{eq:BoundedSpeedStrong} and some facts about standard subspaces that we have collected in Appendix~\ref{appendix:StandardSubspaces} (in particular regarding the modular data $J_H$, $\Delta_H$ of a standard subspace $H$ and Borchers' Theorem), but part~\ref{item:LorentzInvariantSpectrum} can be read independently of this background.

\begin{proposition}\label{proposition:SpectrumMustBeLorentzInvariant}
    Let $E:\Bor(\Rl^d)\to\CP_{\Rl}(\Hil)$ be a spatial localization observable that transforms covariantly under a positive energy representation $U$ of $\Rl^{d+1}$ and obeys the strong speed limit condition \eqref{eq:BoundedSpeedStrong}.
    \begin{enumerate}
        \item\label{item:BorchersRelation} Let $\by\in\Rl^d$ be a unit vector. The modular data of the standard subspace $H(\by):=E(A(\by))\Hil$ satisfy for any $x\in\Rl^{d+1}$ and any $t\in\Rl$
        \begin{align}\label{eq:BoostsOnTranslations}
            \Delta_{H(\by)}^{it} U(x) \Delta_{H(\by)}^{-it}
            =
            U(\La_\by(-2\pi t)x),
            \qquad
            J_{H(\by)}U(x)J_{H(\by)}=U(j_\by x).
        \end{align}
        \item\label{item:LorentzInvariantSpectrum} The energy-momentum spectrum of~$U$ is Lorentz invariant.
    \end{enumerate}
\end{proposition}
\begin{proof}
    \ref{item:BorchersRelation} Consider the lightlike vector $y_\pm:=(\pm1,\by)$. The strong speed limit condition for $E$ and \eqref{eq:AyInvariance2} imply for any $s\geq0$
    \begin{align}
        U(sy_\pm)H(\by)
        =
        U(\pm se_0)H(A(\by)+s\by)
        \subset
        H(A(\by)+s\by+|s|\bB)
        \subset
        H(\by).
    \end{align}
    Since $U$ is a positive energy representation and $y_\pm$ lies in $\overline{V^+}$, the one-parameter group $U(sy_+)$, $s\in\Rl$, has positive generator, and $U(sy_-)$, $s\in\Rl$, has negative generator. Hence Theorem~\ref{theorem:Std}~\ref{item:Borchers} implies $\Delta_{H(\by)}^{it} U(y_\pm) \Delta_{H(\by)}^{-it}=U(e^{\pm 2\pi t}y_\pm)=U(\La_\by(-2\pi t)y_\pm)$, $t\in\Rl$, and $J_{H(\by)}U(y_\pm)J_{H(\by)}=U(-y_\pm)=U(j_\by(y_\pm))$.

    For directions $\bx\in\Rl^d$ that are orthogonal to $\by$, i.e. in the boundary of $A(\by)$, we clearly have $A(\by)+\bx=A(\by)$. In view of the translational covariance of $E$, this implies $U(\bx)H(\by)=H(\by)$, and, by Theorem~\ref{theorem:Std}~\ref{item:AUH}, $\Delta_{H(\by)}^{it} U(\bx) \Delta_{H(\by)}^{-it}=U(\bx)=U(\La_\by(-2\pi t)\bx)$, $t\in\Rl$, and $J_{H(\by)}U(\bx)J_{H(\by)}=U(\bx)$.

    Both results together show that \eqref{eq:BoostsOnTranslations} holds for all $x\in\Rl^{d+1}$.

    \ref{item:LorentzInvariantSpectrum} Since the spectrum of an operator is invariant under conjugation by unitaries, \eqref{eq:BoostsOnTranslations} implies that the energy-momentum spectrum $S$ is invariant under arbitrary Lorentz boosts $\La_{\by}(t)$, $t\in\Rl$, $\by\in\Rl^d$. As any rotation $R\in \SO(d)$ is a product of suitable Lorentz boosts, it follows that $S$ is also rotationally invariant, i.e. Lorentz invariant.
\end{proof}

A consequence of this result is that the strong speed limit condition requires Lorentz symmetry, at least on the level of the spectrum. As an example, consider a representation on $L^2(\Rl^d)$ of the form
\begin{align}
    U(x)=e^{ix_0 w(\Pi_1,\ldots,\Pi_d)}e^{ix_1\Pi_1}\cdots e^{ix_d\Pi_d},
\end{align}
with $\Pi_k$ the canonical momentum operators and $w:\Rl^d\to\Rl_+$ a fixed function (energy-momentum relation). The energy-momentum spectrum of $U$ is the graph of $w$. Typical non-relativistic examples, such as $w(\bp)=\frac{1}{2m}\|\bp\|^2$, are O$(d)$-invariant but not Lorentz invariant and hence do not admit causal spatial localization observables. On the other hand, the representation with $w(\bp)=(m^2+\|\bp\|^2)^{1/2}$, $m\geq0$, implementing a relativistic dispersion relation, has a Lorentz invariant spectrum, so the existence of corresponding causal localization observables is not ruled out by Proposition~\ref{proposition:SpectrumMustBeLorentzInvariant}.

In $d=1+1$ dimensions, one can even use the modular group of the half line $A(1)$ to extend a given translation representation $U$ of $\Rl^2$ to the Poincaré group $\Poi(2)$. For $d>1+1$, the information gathered in \eqref{eq:BoostsOnTranslations} is not enough to produce such an extension, but nonetheless a strong indication that the natural setting to consider is the case in which the translation representation restricts from a representation of $\Poi(d)$. We will therefore from now on restrict ourselves to consider the strongest and most constrained setting of Poincaré covariant spacetime localization observables.

\section{Spacetime localization observables in $\CP_{\Rl}(\Hil)$}\label{section:SpacetimeLocObinPRH}

As motivated above, we now consider Poincaré representations. More precisely, let~$U$ be a strongly continuous representation of the Poincaré group $\Poi(d+1)$ which has positive energy and which is (anti-)unitary in the sense that $U(L)$ is unitary for orthochronous $L$, and anti unitary for anti orthochronous $L$. As is well known, the irreducible representations of this form are classified by mass and spin (helicity) values and can be identified with elementary particles \cite{Wigner:1939}.

We focus on the involution lattice $\CP_{\Rl}(\Hil)$, i.e. we consider spacetime localization observables
\begin{align}
    E:\CC(\Rl^{d+1})\to\CP_{\Rl}(\Hil)
\end{align}
transforming covariantly under $U$. In the preceding section, we have derived some necessary properties of such maps, but it may not be obvious if they actually exist or how to construct examples. The aim is to establish existence and uniqueness results. As in the preceding section, we will make use of standard subspaces (see Appendix~\ref{appendix:StandardSubspaces}).

Let us fix some notation. We will write $H(\CO):=E(\CO)\Hil$ for the ranges of the real projections $E(\CO)$. Since we may restrict $E$ to the spatial translationally covariant (w.r.t. $U|_{\Rl^d}$) causal localization observable $E\circ\cc:\Bor(\Rl^d)\to\CP_{\Rl}(\Hil)$ (Proposition~\ref{prop:SpacetimeObservableInduceCausalSpatialObservables}~\ref{item:SpacetimeToSpatial}), and $E(\cc(A))$ is standard if $A$ and $A^\perp$ have interior points (Proposition~\ref{prop:standardsubspacesarenecessary}), the monotonicity and separation properties of $E$ imply that $E(\CO)$ is standard if $\CO$ and~$\CO'$ have interior points. In this case, we will write $\Delta_{\CO}$, $J_{\CO}$ for its modular data. As a standard subspace, $H(\CO)$ satisfies $H(\CO)=\ker(J_{\CO}\Delta_{\CO}^{1/2}-1)$ (Appendix~\ref{appendix:StandardSubspaces}), so specifying the modular data of $H(\CO)$ is equivalent to specifying $H(\CO)$ itself. This becomes also clear by recalling that the orthogonal projection onto $H(\CO)$ can be expressed as \cite{FiglioliniGuido:1994}
\begin{align}\label{eq:E}
    E(\CO)
    =
    (1+J_{\CO}\Delta_{\CO}^{1/2})(1+\Delta_{\CO})^{-1}.
\end{align}

\subsection{Existence}

In the proof of the Lorentz symmetry of the energy-momentum spectrum (Proposition~\ref{proposition:SpectrumMustBeLorentzInvariant}), we have already made use of the causal completions $\cc(A(\by))$ of the half spaces $A(\by)$ \eqref{eq:TimeZeroHalfSpace}. In our present more general spacetime picture, we define a {\em wedge} to be a subset $W\subset\Rl^{d+1}$ that is a Poincaré transform of $\cc(A(\by))$. It is then clear that for every wedge $W$, there exists a one parameter group $\La_W(t)$, $t\in\Rl$, in the proper orthochronous Lorentz group, and a reflection $j_W$ in the proper Lorentz group, such that
\begin{align}
    \La_W(t)W
    &=
    W,\quad t\in\Rl,\qquad
    j_W(W^\circ)=(W')^\circ,
    \\
    \La_{LW}(t)
    &=
    L\La_W(t)L^{-1},\qquad
    j_{LW}=Lj_W L^{-1},\qquad L\in\Poi(d+1).
\end{align}
In the formula $j_W(W^\circ)=(W')^\circ$ we need the interior because in the logic $\CC(\Rl^{d+1})$, the causally complete wedges are closed, and $j_W(W)\supsetneq W'$. The mismatch between $j_W(W)$ and $W'$ consists however only in the edge of $W$, which is of measure zero. This will play no role for the associated standard subspaces.

In Proposition~\ref{proposition:SpectrumMustBeLorentzInvariant}, we have seen that the modular unitaries $\Delta_W^{it}$ and modular conjugation $J_W$ act on the translations $U(x)$ like the represented Lorentz boosts $U(\La_W(-2\pi t))$ and reflections $U(j_W)$, respectively. This opens up the possibility to start the \emph{definition} of a candidate for a localization observable by requiring that the modular group and conjugation of $H(W)$ are given by $\Delta_{H(W)}^{it}=U(\La_W(-2\pi t))$ and $J_{H(W)}=U(j_W)$, respectively. For the corresponding standard subspace, this means
\begin{align}\label{eq:BGLwedgespace}
    H(W)
    :=
    H_U(W)
    :=
    \ker(U(j_W)e^{-\pi K_W}-1),
\end{align}
where $W$ is any wedge and $U(\La_W(t))=e^{itK_W}$, i.e. $K_W$ is the generator of the Lorentz boosts in $W$-direction in the representation $U$. This yields a closed real linear space with associated projection $E_U(W)\in\CP_{\Rl}(\Hil)$, and clearly $E_U(W)$ is defined entirely in terms of the representation~$U$.

This idea goes back to Brunetti, Guido and Longo under the name of \emph{modular localization} \cite{BrunettiGuidoLongo:2002}, and we here briefly summarize their main results.

\begin{theorem}\emph{\cite{BrunettiGuidoLongo:2002}}\label{theorem:BGLwedges}
    Let $U$ be a positive energy representation of $\Poi(d+1)$, and $W\subset\Rl^{d+1}$ a wedge.
    \begin{enumerate}
        \item\label{item:BGLwedgeduality} $H_U(W)$ is a standard subspace with $H_U(\overline{W'})=H_U(W)'$.
        \item\label{item:BGLfactoriality} If $U$ does not contain the trivial representation, then $H_U(W)\wedge H_U(W)'=0$.
        \item\label {item:BGLwedgeisotony} If $\tilde W$ is another wedge such that $\tilde W\subset W$, then $H_U(\tilde W)\subset H_U(W)$.
        \item\label{item:BGLcovariance} For any $L\in\Poi(d+1)$, we have $U(L)H_U(W)=H_U(LW)$.
    \end{enumerate}
\end{theorem}

\begin{definition}\label{def:BGLmap}
    Let $U$ be a positive energy representation of $\Poi(d+1)$. Then the {\em Brunetti-Guido-Longo (BGL)} map is
    \begin{align}
        H_U:\;&\CC(\Rl^{d+1})\to\CP_{\Rl}(\Hil),
        \\
        H_U(\CO)
        &:=
        \bigvee_{C\in\CC(\Rl^{d+1})\;\text{convex} \atop{C\subset\CO}}\bigwedge_{\CW\ni W\supset C}H_U(W),
        \label{eq:BGLgeneral}
    \end{align}
\end{definition}

To unpack this definition, we first consider a convex causally complete set $C\in\CC(\Rl^{d+1})$. In this case, Definition~\ref{def:BGLmap} simplifies to $H_U(C)=\bigwedge_{\CW\ni W\supset C}H_U(W)$ because $C\mapsto H_U(C)$ is clearly monotone. As a causally complete set $C\in\CC(\Rl^{d+1})$ is convex if and only if it coincides with the intersection of all wedges containing it, $C=\bigcap_{\CW\ni W\supset C}W$ \cite{ThomasWichmann:1997}, we in this case have the natural definition
\begin{align}
    H_U\left(\bigwedge_{\CW\ni W\supset C}W\right)
    =
    \bigwedge_{\CW\ni W\supset C}H_U(W).
\end{align}
As also $\CW\ni W\mapsto H_U(W)$ is monotone by Theorem~\ref{theorem:BGLwedges}~\ref{item:BGLwedgeisotony}, we note that this implies in particular that Definition~\ref{def:BGLmap} is consistent with \eqref{eq:BGLwedgespace}. For non-convex regions $\CO$, the BGL map defines $H_U(\CO)$ to be generated by all $H_U(C)$ with $C\in\CC(\Rl^{d+1})$ convex and contained in $\CO$.

It is now easy to verify that $H_U$ is a causal spacetime localization observable.

\begin{theorem}
    Let $U$ be a positive energy representation of $\Poi(d+1)$ that does not contain the trivial representation. The BGL map (Definition~\ref{def:BGLmap}) $H_U$ is a $U$-covariant spacetime localization observable.
\end{theorem}
\begin{proof}
    We begin by checking the normalization of $H_U$. By Definition~\ref{def:BGLmap},
    \begin{align}
        H_U(\emptyset)=\bigwedge_{W\in\CW}H_U(W)\subset H_U(W_0)\wedge H_U(\overline{W_0'})
    \end{align}
    for any wedge $W_0$. According to Theorem~\ref{theorem:BGLwedges}~\ref{item:BGLwedgeduality} and \ref{item:BGLfactoriality}, $H_U(W_0)\wedge H_U(\overline{W_0'})=0$, so $H_U(\emptyset)=0$ follows. Since $\Rl^{d+1}$ is convex but not contained in any wedge, $H_U(\Rl^{d+1})=1$ holds by definition.

    Given two convex sets $C_1,C_2\in\CC(\Rl^{d+1})$ such that $C_1\subset C_2'$, there exists a wedge $W$ such that $C_1\subset W$ and $C_2\subset\overline{W'}$ \cite{ThomasWichmann:1997}. This implies $H_U(C_1)\subset H_U(W)=H_U(\overline{W'})'\subset H_U(C_2)'$, so $H_U$ preserves separation.

    Regarding $\sigma$-additivity, let $(\CO_n)\subset\CC(\Rl^{d+1})$ be a countable separated family. Then any convex $C\in\CC(\Rl^{d+1})$ with $C\subset\bigvee_n\CO_n$ lies in only one of the $\CO_n$'s. This implies the additivity of $H_U$ by \eqref{eq:BGLgeneral}, and we conclude that $H_U$ is an observable in the sense of Definition~\ref{def:Observable}.

    The $U$-covariance of $H_U$ follows from Theorem~\ref{theorem:BGLwedges}~\ref{item:BGLcovariance}.
\end{proof}

The spaces $H_U(\CO)$ are closely connected to free quantum field theories, which also gives a more direct understanding of the vectors in them. For simplicity, let us consider the irreducible representation $U$ of mass $m>0$ and spin $0$, as in Example~\ref{example:NewtonWigner} (see that example for the definition of the representation space and the dispersion relation $\om$). This is a representation of the full Poincaré group (including in particular the time reflection $T(x)=(-x_0,\bx)$).

To describe $H_U(\CO)$, we consider $\CO_R=\cc(B_R(0))$ to be the causal completion of a ball of radius $R>0$, and split $\psi\in\Hil$ into time reflection invariant components
\begin{align}\label{eq:KGcomponents}
    \psi_+=\frac12(\psi+U(T)\psi),\qquad
    \psi_-=\frac{1}{2i\om}(\psi-U(T)\psi);
\end{align}
these are related to the Cauchy data of the Klein-Gordon equation. Any $\psi\in\Hil$ is in particular a tempered distribution on $\Rl^d$, and we consider the inverse transforms $\EuScript{F}^{-1}(\psi_\pm)$. Then
\begin{align}\label{eq:HUconcrete}
    H_U(\cc(B_R(0)))
    &=
    \{\psi\in\Hil\col \supp\EuScript{F}^{-1}\psi_\pm\subset B_R(0)\}
    \\
    &=
    \overline{\{\bp\mapsto\tilde f(\om(\bp),\bp) \col f\in C_{c,\Rl}^\infty(\cc(B_R(0))\}}
    .
\end{align}
For a proof of this fact, see \cite[App.~A]{LechnerLongo:2014} (the arguments given there relate to $d=1$, but generalize to $d>1$).

The characterizations \eqref{eq:HUconcrete} provide us with many examples of vectors in $H_U(\CO_R))$ (and by covariance and additivity, in the spaces $H_U(\CO)$ for general $\CO\in\CC(\Rl^{d+1})$). However, explicit expressions for the projections \eqref{eq:E} onto these spaces are not known because the modular data $\Delta_{H_U(\CO)}$, $J_{H_U(\CO)}$ are not known explicitly unless $\CO$ is a wedge. In the case of wedges, closed formulas for the projections can be derived (see \cite[Proposition~6.1 a)]{CorreadaSilvaLechner:2025} for the two-dimensional case). A further discussion of these spaces will be given in Section~\ref{section:Probability}.

\subsection{Uniqueness}\label{section:uniqueness} Having established the existence of spacetime localization observables $E:\CC(\Rl^{d+1})\to\CP_{\Rl}(\Hil)$ with the example of the BGL maps $E=E_U$, we next consider uniqueness questions. The definition of $E_U$ (Definition~\ref{def:BGLmap}) is based on three assumptions:
\begin{enumerate}[(U1)]
    \item\label{eq:U1} For every wedge $W\in\CW$, the modular data of $H(W)$ are given by the boosts in $W$-direction and the reflection $j_W$ at the edge of $W$ as represented by $U$ ({\em Bisognano-Wichmann property}).

    \item\label{eq:U2} For a convex region $C=\bigcap_{\CW\ni W\supset C}W$, one has $H(C)=\bigwedge_{\CW\ni W\supset C}H(W)$.

    \item\label{eq:U3} For a general region $\CO\in\CC(\Rl^{d+1})$, additivity holds in the form that  $H(\CO)$ is the smallest closed real subspace containing all $H(C)$, where $C\in\CC(\Rl^{d+1})$ runs over all convex subsets of $\CO$.
\end{enumerate}
It is clear that $H_U$ is the unique localization observable satisfying these three requirements.

To discuss these three assumptions, recall that an observable is always monotone (Lemma~\ref{lemma:Observables}). Given a convex causally complete region $C=\bigcap_{\CW\ni W\supset C}W$, any observable $H:\CC(\Rl^{d+1})\to\CP_{\Rl}(\Hil)$ must therefore satisfy $H(C)\subset \bigwedge_{\CW\ni W\supset C}H(W)$. Hence assumption (U2) represents the maximal choice for $H(C)$ once the wedge spaces $H(W)$ are fixed. Although no classification of smaller choices for $H(C)$ exists, this maximality is a natural assumption which appears frequently in similar situations \cite{SummersWichmann:1987,Lechner:AQFT-book:2015,MorinelliNeeb:2023}.

Assumption \ref{eq:U3} expresses a form of additivity of $H$ that goes beyond the additivity required in Definition~\ref{def:Observable} because the convex subsets $C\subset\CO$ do not have to be separated. By monotonicity, any localization observable must satisfy $H(\CO)\supset H(C)$ for any convex $C\subset\CO$. Hence \ref{eq:U3} amounts to taking the minimal extension of an observable defined on convex causally complete region to general $\CO\in\CC(\Rl^{d+1})$. This again seems to be a reasonable assumption, which in addition holds automatically in restriction to the Boolean sublogic $\Bor(\Rl^d)$ (Lemma~\ref{lemma:Observables}~\ref{item:fpreservesjoins}).

Accepting \ref{eq:U2} and \ref{eq:U3}, uniqueness of spacetime localization observables is directly linked to the Bisognano-Wichmann assumption \ref{eq:U1} which might not appear natural from a quantum-mechanical perspective. It has its origin in quantum field theory and is known to hold in any Wightman QFT \cite{BisognanoWichmann:1976}. In our setting of localization observables, we may of course not argue via results in quantum field theory. Nonetheless, we have seen in Proposition~\ref{proposition:SpectrumMustBeLorentzInvariant} that thanks to Borchers' Theorem, \ref{eq:U1} partially follows from our assumptions. To explain this point further, we consider a simplified situation.

\begin{lemma}
    Let $E:\CC(\Rl^{d+1})\to\CP_{\Rl}(\Hil)$ be a $U$-covariant spacetime localization observable, and $W$ a wedge. Then $U(\La_W(2\pi t))\Delta_W^{it}$ and $U(j_W)J_W$ commute with $U(x,L)$ for all $x\in\Rl^{d+1}$ and all Lorentz transformations $L$ satisfying $LW=W$.

    If the restriction of $U$ to the proper Poincaré group is irreducible and the dimension is $1+1$, then
    \begin{align}
        \Delta_W^{it}=U(\La_W(-2\pi t)), \qquad J_W=\pm U(j_W),
        \\
        H(W)=H_U(W)\quad\text{or}\quad
        H(W)=iH_U(W).
    \end{align}
\end{lemma}
\begin{proof}
    By the same proof as in Proposition~\ref{proposition:SpectrumMustBeLorentzInvariant}~\ref{item:BorchersRelation}, $U(\La_W(-2\pi t))$ and $\Delta_W^{it}$ act in the same manner on translations, and the same is true for $U(j_W)$ and $J_W$. As all operators entering here also commute with $U(0,L)$ for Lorentz transformations $L$ satisfying $LW=W$, the first claim follows.

    In $d=1+1$, the proper Poincaré group is generated by translations, boosts $\La_W(t)$ for the Rindler wedge $W$, and the reflection $j_W$. Hence irreducibility implies that $U(\La_W(t))\Delta_W^{-it}$ and $U(j_W)J_W$ are multiples of the identity, $U(\La_W(t))\Delta_W^{-it}=z(t)\cdot1$ and $U(j_W)J_W=z'\cdot1$ for some $z(t),z'\in\Cl$, $|z(t)|=|z'|=1$. Both these operators map $H(W)$ onto itself, and since this space is standard, we have $z(t),z'\in\{1,-1\}$. As $z(t)$ is a one-parameter group, we conclude $U(\La_W(t))\Delta_W^{-it}=1$ and $J_W=\pm U(j_W)$. As $J_{iH(W)}=-J_{H(W)}$, this implies the claimed form of~$H(W)$.
\end{proof}

This argument shows that at least in the simple situation of an irreducible representation in $1+1$ dimensions, \ref{eq:U1} almost follows (up to the freedom of a factor $i$). With more involved arguments, Mund has shown that in case $U$ is a finite direct sum of irreducible massive representations and the dimension is $1+3$, \ref{eq:U1} follows \cite[Theorem~5]{Mund:2001}. Another argument, due to Morinelli, shows that \ref{eq:U1} also follows if a $U$ is a general scalar representation in dimension $1+3$ \cite[Theorem~4.4]{Morinelli:2018}. On the other hand, also counterexamples to the Bisognano-Wichmann property are known \cite{Yngvason:1994,Borchers:1998_2,Morinelli:2018}, albeit only quite artificial ones. We therefore conclude that up to rather artificial counterexamples, the natural conditions \ref{eq:U2} and \ref{eq:U3} fix the BGL map as the unique spacetime localization observable transforming covariantly under a given representation $U$.

\medskip

To illustrate the properties of the localization observable $H_U$, we consider the example of the irreducible positive energy representation $U=U_{m,0}$ of mass $m>0$ and spin zero, and show that $E_{U_{m,0}}$ is incompatible with the Newton-Wigner localization observable (projection-valued measure) $P^{\NW}_{U_{m,0}}:\Bor(\Rl^d)\to\CP(\Hil)$ from Example~\ref{example:NewtonWigner}.

\begin{proposition}\label{prop:NWandEareIncompatible}
    Let $A,B\in\Bor(\Rl^d)$ be bounded. Then
    \begin{align}
        E_{U_{m,0}}(\cc(A)) \wedge P^{\NW}_{U_{m,0}}(B) = 0.
    \end{align}
\end{proposition}
\begin{proof}
    Without loss of generality, we may assume $A=B$ to be the ball $B_R(0)$ of radius $R>0$ around the origin. Taking into account the characterizations \eqref{eq:HUconcrete} and \eqref{eq:NWPVM} of $E_{U_m}$ and $P^{\NW}_{U_{m,0}}$, respectively, a function $\psi$ in the range of $E_{U_{m,0}}(W) \wedge P^{\NW}_{U_{m,0}}(B)$ can be represented as
    \begin{align}
        \psi=\psi_++i\om\psi_-=\om^{1/2}\tilde f,
    \end{align}
    with $f\in L^2(\Rl,d\bx)$, $\supp f\subset B_R(0)$ and $\psi_+=\frac12(\psi+U_{m,0}(T)\psi)$, $\psi_-=\frac{1}{2i\om}(\psi-U_{m,0}(T)\psi)$ with $\supp\CF^{-1}\psi_\pm\subset B_R(0)$. As time reflection is represented as $(U_{m,0}(T)\psi)(\bp)=\overline{\psi(-\bp)}$, we see $\psi_+=(\om^{1/2}\tilde f)_+=\frac12(\psi+U_{m,0}(T)\psi)=\om^{1/2}\widetilde{\real(f)}$ and $\psi_-=\om^{-1/2}\widetilde{\imag(f)}$. On the Fourier transformed side, $\om^{\pm1/2}$ acts as the fractional power $(m^2-D)^{\pm 1/2}$ (with $D$ the Laplace operator), and this operator is known to be anti-local \cite{Murata:1973}. This implies that for a function~$g$ of compact support, $(m^2-D)^{\pm 1/2}g$ has compact support if and only if $g=0$. In view of the above representations for $\psi_\pm$, we arrive at $\psi=0$.
\end{proof}

\section{Probabilistic Interpretation}\label{section:Probability}

Having established the existence of covariant spacetime localization observables in~$\CP_{\Rl}(\Hil)$ in wide generality, we now turn to their probabilistic interpretation. A basic concept for this is a probability measure on an involution lattice.

\begin{definition}\label{def:ProbabilityMeasureOnInvolutionLattice}
    A {\em probability measure} on an involution lattice $\CQ$ is a map $\mu:\CQ\to[0,1]$ satisfying
    \begin{enumerate}
        \item $\mu(0)=0$, $\mu(1)=1$ (normalization),
        \item $A\leq B\Rightarrow \mu(A)\leq \mu(B)$ (monotonicity),
        \item\label{item:MeasureAdditivity} For any countable separated subset $\CQ_0\subset\CQ$,
        \begin{align}\label{eq:MeasureAdditivity}
            \mu\left(\bigvee_{A\in\CQ_0} A\right)
            =
            \sum_{A\in\CQ_0} \mu(A).
        \end{align}
    \end{enumerate}
    A {\em finitely additive probability measure} is defined in the same way, but with \eqref{eq:MeasureAdditivity} only required to hold for finite separated sets $\CQ_0$.

    The set of all finitely additive probability measures on $\CQ$ is denoted $\Prob_0(\CQ)$.
\end{definition}

When $\CQ$ is a logic and $\mu\in\Prob_0(\CQ)$, we have $\mu(A^\perp)+\mu(A)=1$ for all $A\in\CQ$ because $A\vee A^\perp=1$. Furthermore, for logics monotonicity follows from additivity and positivity because for $A\leq B$, we have $B=A\vee(A^\perp\wedge B)$ and hence $\mu(B)=\mu(A)+\mu(A^\perp\wedge B)\geq\mu(A)$. In general involution lattices, both these properties might fail.

The (finitely additive) probability measures on $\CQ=\CP(\Hil)$ are described by Gleason's Theorem. Writing $S(B(\Hil))$ for the state space of the $C^*$-algebra $B(\Hil)$, this is the following statement (see \cite[Ch.~3]{Hamhalter:2003},\cite[Theorem~4.2.9]{Landsman:2017}).

\begin{theorem}\label{theorem:GleasonOrthogonal}
    Let $\Hil$ be a Hilbert space of dimension $\dim\Hil>2$. Then
    \begin{align}\label{eq:muom}
        S(B(\Hil)) \to \Prob_0(\CP(\Hil)),\qquad
        \om\mapsto\mu:=\om|_{\CP(\Hil)}
    \end{align}
    is a bijection, and $\mu=\om|_{\CP(\Hil)}$ is a probability measure (i.e., $\sigma$-additive) if and only if~$\om$ is normal.
\end{theorem}

Given an observable $P:\Bor(\Rl^d)\to\CP(\Hil)$ taking values in complex projections, and a probability measure $\mu=\om|_{\CP(\Hil)}$, $\om\in\CB(\Hil)_*$, the map
\begin{align}\label{eq:Pom}
    \mu_P
    :=
    \om\circ P:\Bor(\Rl^d)\to[0,1],\qquad P\mapsto\om(P(A)),
\end{align}
is a probability measure in the usual sense of measure theory. In the context of localization observables, $\om(P(A))$ is interpreted as the probability of detecting a particle described by $\om$ in the Borel set~$A$.

In contrast to this familiar description of the Born rule, the situation for the involution lattice $\CP_{\Rl}(\Hil)$ looks very different.

\begin{theorem}\label{theorem:SymplecticGleasonTheoremHilbert}
    Let $\Hil$ be an infinite-dimensional Hilbert space.
    \begin{enumerate}
        \item\label{item:FiniteMeasureHilbert} Any finitely additive probability measure $\mu$ on $\CP_{\Rl}(\Hil)$ vanishes on all finite-dimensional real subspaces of $\Hil$.
        \item\label{item:NoInfiniteMeasureHilbert} If $\Hil$ is separable, there exists no probability measure $\mu$ on $\CP_{\Rl}(\Hil)$.
    \end{enumerate}
\end{theorem}
\begin{proof}
    \ref{item:FiniteMeasureHilbert} is proven in the more general context of symplectic spaces in Appendix~\ref{section:SymplecticGleason} (Theorem~\ref{theorem:SymplecticGleasonTheorem}), where also the finite-dimensional case is treated for completeness.

    \ref{item:NoInfiniteMeasureHilbert} As $\Hil$ is separable, it has a countable orthonormal basis $(e_n)_{n\in\Nl}$. The closed subspaces $H_n:=\Cl e_n$ are pairwise orthogonal, so in particular $H_n\subset H_m'$ for $n\neq m$. Furthermore, $\bigvee_nH_n=\Hil$. As any probability measure $\mu$ on $\CP_{\Rl}(\Hil)$ satisfies $\mu(H_n)=0$ for all $n$ by part~\ref{item:FiniteMeasureHilbert}, $\sigma$-additivity and normalization of $\mu$ imply the contradiction
    \begin{align*}
        1=\mu(\Hil)=\mu(\bigvee_nH_n)=\sum_n\mu(H_n)=0.
    \end{align*}
    Hence $\mu$ does not exist.
\end{proof}

In comparison to the ``orthogonal'' Gleason Theorem~\ref{theorem:GleasonOrthogonal} for $\CP(\Hil)$, Theorem~\ref{theorem:SymplecticGleasonTheoremHilbert} has the form of a No-Go result: Disregarding the case of non-separable Hilbert spaces, or the case of finitely additive measures on separable ones, there exist no probability measures on $\CP_{\Rl}(\Hil)$. So there is no universal way to obtain probability measures on general observables $E:\CL\to\CP_{\Rl}(\Hil)$ taking values in the involution lattice~$\CP_{\Rl}(\Hil)$.

Our arguments so far do not rule out that for specific observables $E:\CL\to\CP_{\Rl}(\Hil)$, there might exist probability measures on the involution sublattice $E(\CL)\subset\CP_{\Rl}(\Hil)$. However, also probability measures on these smaller lattices are severely restricted once standard subspaces come into play. To explain this, we compare the supremum $E_1\vee E_2$ of two projections with their sum $E_1+E_2$, resembling the finite additivity property of probability measures. So we have reason to consider operators of the form
\begin{align}\label{eq:X}
    X(E_1,E_2)
    :=E_1+E_2-E_1\vee E_2,\qquad
    E_1,E_2\in\CP_{\Rl}(\Hil).
\end{align}
Note that in case $E\wedge E'=0$, the operator $X(E,E')=E'-E^\perpr$ measures the mismatch between the symplectic and real orthogonal complement.

\begin{proposition}\label{proposition:X}
    Let $H_1,H_2\in\CP_{\Rl}(\Hil)$ be closed real subspaces and $E_k:=E_{H_k}$, $k=1,2$, the projections onto them.
    \begin{enumerate}
        \item\label{item:Xnorm} $X(E_1,E_2)$ is a real selfadjoint contraction with norm $\|X(E_1,E_2)\|=\|E_1E_2\|$.

        \item\label{item:Xortho} $X(E_1,E_2)=0$ if and only if $E_1\leq E_2^\perpr$.

        \item\label{item:Xstandard} If $E_1$ is standard and $E_1\leq E_2'$, then $X(E_1,E_2)=0$ if and only if $E_2=0$.
    \end{enumerate}
\end{proposition}
\begin{proof}
    We abbreviate $X:=X(E_1,E_2)$ in this proof.

    \ref{item:Xnorm} We first consider the special case in which $E_1\wedge E_2=0$ and $E_1\vee E_2=1$ holds, i.e. $H_1$ and $H_2$ intersect trivially and span a dense subspace of $\Hil$. Then the norm of the bounded operator $X$ coincides with the norm of its restriction to $H_1+H_2$. For $h_1\in H_1$, $h_2\in H_2$, we find $X(h_1+h_2)=E_1h_2+E_2h_1$ and
    \begin{align}\label{eq:X2}
        X^2(h_1+h_2)=E_1E_2E_1h_1+E_2E_1E_2h_2,
    \end{align}
    i.e. $X^2$ can be identified with $\begin{pmatrix}E_1E_2E_1 & 0\\0&E_2E_1E_2\end{pmatrix}$ on the (non-orthogonal) direct sum $H_1+H_2$. As $X^2$ is selfadjoint, its norm coincides with its spectral radius $r(X^2)$, which is
    \begin{align}
        \|X\|^2
        &=
        r(X^2)=r\begin{pmatrix}E_1E_2E_1 & 0\\0&E_2E_1E_2\end{pmatrix}
        \\
        &=
        \max\{r(E_1E_2E_1),r(E_2E_1E_2)\}
        =
        \max\{\|E_1E_2E_1\|,\|E_2E_1E_2\|\}
        =
        \|E_1E_2\|^2.
    \end{align}
    Hence $\|X\|=\|E_1E_2\|$.

    For general position of $H_1,H_2$, we use the Halmos decomposition \cite{Halmos:1969} of $\Hil$ w.r.t. $E_1,E_2$,
    \begin{align*}
        \Hil
        =
        (H_1\wedge H_2)\oplus_{\Rl}
        (H_1\wedge H_2^\perpr)\oplus_{\Rl}
        (H_1^\perpr\wedge H_2)\oplus_{\Rl}
        (H_1\vee H_2)^\perpr\oplus_{\Rl}
        \CK,
    \end{align*}
    where $\CK$ is defined to be the real orthogonal complement of the direct sum of the first four terms. In this decomposition, we have
    \begin{align*}
        E_1
        &=
        1\oplus 1\oplus 0\oplus 0\oplus F_1
        ,\\
        E_2
        &=
        1\oplus 0\oplus 1\oplus 0\oplus F_2
        ,\\
        E_1\vee E_2
        &=
        1\oplus 1\oplus 1\oplus 0\oplus 1
        ,\\
        X
        &=
        1\oplus 0\oplus 0\oplus 0\oplus (F_1+F_2-1),
    \end{align*}
    with the real orthogonal projections $F_i:=E_i|_{\CK}\in\CP_{\Rl}(\CK)$. Since $F_1\vee F_2=1_{\CK}$ and $F_1\wedge F_2=0$, the first part of this proof implies $\|F_1+F_2-1\|=\|F_1F_2\|$. From the Halmos decomposition it is clear that $\|X\|=1$ in case $H_1\cap H_2\neq\{0\}$ and $\|X\|=\|F_1F_2\|$ otherwise, which coincides with $\|E_1E_2\|$ in both cases.

    \ref{item:Xortho} follows directly from \ref{item:Xnorm} and $E_1E_2=0\Leftrightarrow E_1\leq E_2^\perpr$.

    \ref{item:Xstandard} In \ref{item:Xortho} we saw that $X=0$ implies $E_2\leq E_1^\perpr$. If also $E_2\leq E_1'$, then the range of $E_2$ is complex orthogonal to the complex span of the range of $E_1$. In case that $E_1$ is standard, this complex span is dense, so $E_2=0$ follows. The other direction is obvious.
\end{proof}

Given any normal state $\om$ on $\CB(\Hil)$, the real linear analogue of the probability measure \eqref{eq:muom} from Gleason's Theorem,
\begin{align}
    \real\om:\CP_{\Rl}(\Hil)\to[0,1],\qquad E\mapsto\real\om(E),
\end{align}
is normalized and monotone but fails to be $\sigma$-additive (at least for separable $\Hil$) by Theorem~\ref{theorem:SymplecticGleasonTheoremHilbert}~\ref{item:NoInfiniteMeasureHilbert}. Given an observable $E:\CL\to\CP_{\Rl}(\Hil)$ mapping a logic $\CL$ into the involution lattice $\CP_{\Rl}(\Hil)$ such that there exist $A,B\in\CL$, $A\leq B^\perp$, with $E(A)$ standard and $E(B)\neq0$ (which is the case for our localization observables), we have $X(E(A),E(B))\neq0$ by Proposition~\ref{proposition:X}. Hence there exists $\om\in\CB(\Hil)_*$ such that
\begin{align}\label{eq:realom}
    \real\om(E(A\vee B)) \neq \real\om(E(A))+\real\om(E(B)),
\end{align}
i.e. $\real\om|_{E(\CL)}$ is not even finitely additive.

\bigskip

This discussion shows that no natural probability measures on observables mapping into $\CP_{\Rl}(\Hil)$ exist. However, despite this failure of strict additivity, it holds approximately in relevant situations. To show this, we will derive a cluster theorem for standard subspaces which can be proven in a similar manner as Fredenhagen's cluster theorem in quantum field theory. For a concise presentation, we extract the following fact from the proof of the theorem in \cite{Fredenhagen:1985}.

\begin{lemma}\label{lemma:KlausCluster}
    Let $m,\delta,C>0$ and $\Cl_\delta:=\Cl\setminus((-\infty,-\delta]\cup[\delta,\infty))$. Then any analytic function $f:\Cl_\delta\to\Cl$ that satisfies $|f(z)|\leq Ce^{-m|\imag z|}$ for all $z\in\Cl_\delta$ satisfies also $|f(0)|\leq C e^{-m\delta}$.
\end{lemma}

\begin{theorem}\label{thm:MassGapClustering}
    Let $E,F\in\CP_{\Rl}(\Hil)$ and let $V(t)=e^{itQ}$ be a unitary one-parameter group with generator $Q$ satisfying $Q\geq m\cdot1$ for some $m>0$. If there exists $\delta>0$ such that $E\leq V(t)F'V(t)^*$ for all $t\in[-\delta,\delta]$, then
    \begin{enumerate}
        \item\label{item:NormBound1} $\|EF\| \leq e^{-m\delta}$,
        \item\label{item:TrivialIntersection} $E\wedge F=0$,
        \item\label{item:NormBound2} $\|E+F-E\vee F\|\leq e^{-m\delta}$.
    \end{enumerate}
\end{theorem}
\begin{proof}
    \ref{item:NormBound1} Let $\psi,\xi\in\Hil$ and consider the function
    \begin{align*}
        f:\Rl\to\Cl,\qquad
        f(t)
        &:=
        \langle E\psi,e^{itQ}F\xi\rangle
        .
    \end{align*}
    Since $Q$ is positive, $f$ has a bounded analytic continuation to the upper half plane. Furthermore, our assumption $E\leq e^{itQ}F'e^{-itQ}$ for $|t|\leq\delta$ implies $e^{itQ}F\xi\in(E\Hil)'$, namely $f(t)$ is real for $t\in[-\delta,\delta]$. Hence the Schwarz reflection principle can be used to conclude that~$f$ has an analytic continuation to the cut plane $\Cl_\delta:=\Cl\setminus((-\infty,-\delta]\cup[\delta,\infty))$, and the spectral assumption on $Q$ implies the bound $|f(z)|\leq C e^{-m|\imag z|}$, $z\in\Cl_\delta$, where $C:=\|E\psi\|\|F\xi\|$. We are therefore in the situation of Lemma~\ref{lemma:KlausCluster}, which yields
    \begin{align*}
        |f(0)|
        =
        |\langle E\psi,F\xi\rangle|
        \leq
        e^{-m\delta}\|E\psi\| \|F\xi\|.
    \end{align*}
    To derive the claimed norm bound from this, we estimate
    \begin{align*}
        \|EF\|
        &=
        \sup_{\|\psi\|,\|\xi\|\leq1}|\real\langle \psi,EF\xi\rangle|
        \\
        &=
        \sup_{\|\psi\|,\|\xi\|\leq1}|\real\langle E\psi,F\xi\rangle|
        \\
        &=
        \sup_{\|\psi\|,\|\xi\|\leq1}|\langle E\psi,F\xi\rangle|\qquad\qquad\text{(because\;} E\leq F'\text{)}
        \\
        &\leq
        e^{-m\delta}\sup_{\|\psi\|,\|\xi\|\leq1}\|E\psi\| \|F\xi\|
        \\
        &=
        e^{-m\delta}.
    \end{align*}

    \ref{item:TrivialIntersection} In part \ref{item:NormBound1}, we have in particular shown $\|EF\|<1$, so that $(EF)^n\to0$ in operator norm as $n\to\infty$. This implies $E\wedge F=\slim_n(EF)^n=0$.

    \ref{item:NormBound2} follows immediately from \ref{item:NormBound1} and Proposition~\ref{proposition:X}~\ref{item:Xnorm}.
\end{proof}

Note that for $F=E'$, this theorem says that under the stated assumptions, $E'$ almost coincides with $E^\perpr$ in the sense that $X(E,E')=E'-E^\perpr$ has small norm.

We summarize the properties of the maps \eqref{eq:realom} as fuzzy measures (see, for example, \cite{Denneberg:1994}) in the case of a vector state. The generalization to a normal state $\om\in B(\Hil)_*$ is straightforward.

\begin{corollary}
    Let $U$ be a positive energy representation of the Poincaré group on a Hilbert space $\Hil$ with mass gap\footnote{That is, the Hamiltonian $Q$ satisfies $\sigma(Q)\geq m$.} $m>0$, and $E_U\circ\cc:\Bor(\Rl^d)\to\CP_{\Rl}(\Hil)$ the corresponding spatial localization observable. Then, for any a unit vector $\psi\in\Hil$, the map
    \begin{align*}
        \mu_{U,\psi} : \Bor(\Rl^d)\to\Rl,\qquad A\mapsto\|E_U(\cc(A))\psi\|^2
    \end{align*}
    is a fuzzy approximately additive measure: It maps into $[0,1]$, is normalized as $\mu_{U,\psi}(\emptyset)=0$, $\mu_{U,\psi}(\Rl^d)=1$, is monotone, and approximately additive in the sense
    \begin{align}
        |\mu_{U,\psi}(A\vee B)-\mu_{U,\psi}(A)-\mu_{U,\psi}(B)|
        \leq
        e^{-md(A,B)}.
    \end{align}
\end{corollary}
\begin{proof}
    Normalization and monotonicity follow because $E_U$ is an observable, and the range property is a direct consequence. For showing approximate additivity, we note that $\cc(A)+te_0$ lies non-timelike to $\cc(B)$ for $|t|<d(A,B)$. Hence covariance and separation preservation imply $E_U(\cc(A))\leq U(t)E_U(\cc(B))'\,U(-t)$ for $|t|<d(A,B)$, where $U(t)$ denotes the time translations, and Theorem~\ref{thm:MassGapClustering} implies $\|E_U(\cc(A\vee B))-E_U(\cc(A))-E_U(\cc(B))\|\leq e^{-md(A,B)}$. Evaluation in $\real\langle\psi,\,\cdot\,\psi\rangle$ gives the claimed result.
\end{proof}

For spatial localization regions $A,B$ that are a distance $d>0$ apart such that $md\gg1$ (i.e. $d$ is large compared to the Compton wavelength corresponding to mass $m$), this result says that $\mu_{U,\psi}$ is additive up to a negligible error. On scales relevant for quantum field theoretic effects, additivity however fails.

It is also meaningful to compare the modular localization observable with the Newton-Wigner localization observable (Example~\ref{example:NewtonWigner}). As such a comparison is essentially known, we only sketch the ideas: Given $\psi=\psi_++i\om\psi_-\in E_{U_{m,0}}(B)\Hil$, where $B$ is a ball of radius $R>0$, we know $P^{\NW}_{U_{m,0}}(B)^\perp\psi\neq0$ by the incompatibility of the two localization observables (Proposition~\ref{prop:NWandEareIncompatible}). However, if the ball $B$ entering the Newton-Wigner projection is enlarged to a ball $B_\delta$ of radius $R+\delta>R$, then $\|P^{\NW}_{U_{m,0}}(B_\delta)^\perp\psi\|$ can be shown to decay exponentially as $\delta\to\infty$. Hence $\psi$ is not sharply Newton-Wigner localized in $B$, but up to an exponentially small error in $\delta$, it is Newton-Wigner localized in $B_\delta$. The proof of this can be carried out by estimating the spread of compactly supported functions induced by the  operators $(m^2-D)^{\pm1/2}$ (with $D$ the Laplace operator), see Proposition~4.3 and Theorem~4.5 of \cite{LechnerSanders:2016}.

These observations show that the modular localization observables $E_U(B)$ differs from the Newton-Wigner localization observable $P^{\NW}_U(B_\delta)$ only by errors that can be made arbitrily small by suitably extending the localization region.

\medskip

As the first moment of the POVM localization observable introduced by Moretti \cite{Moretti:2023} and its generalizations \cite{DeRosaMoretti:2024} have the Newton-Wigner operators as their first moment, this observation suggests also an at least approximate relation between the approach presented here and localization observables based on POVMs. In this respect, it is interesting to note that fuzzy measures like $\mu_{U,\psi}$ define a Choquet integral and are frequently used in quantum theory (see, for example, \cite{Sorkin:1994,Cerreia-VioglioMaccheroniMarinacciMontrucchio:2018,CarcassiAidalaThrien:2026}). An investigation of these possible connections is left for future work.

\appendix

\section{Standard subspaces}\label{appendix:StandardSubspaces}

In this appendix we gather some background material on standard subspaces that we use in Sections~\ref{subsection:CausalLocalizationObservablesPRH} and Section~\ref{section:SpacetimeLocObinPRH}. For more detailed information, we refer to \cite{Longo:2008_2} and \cite{CorreadaSilvaLechnerLongo:2026}.

By definition, a standard subspace of a complex Hilbert space $\Hil$ is a closed real subspace $H\subset\Hil$ that is separating in the sense $H\cap iH=\{0\}$ and cyclic in the sense that $H+iH$ is dense in $\Hil$. The symplectic complement $H'$ of $H$ (see Example~\ref{example:PRH}) is standard if and only if $H$ is standard. The set $\Std(\Hil)$ of all standard subspaces of $\Hil$ is not a lattice because the intersection of two standard subspaces can easily fail to be standard, e.g. the intersection of $H$ and $iH$ is $\{0\}$.

The {\em Tomita operator} of a standard subspace $H$ is the antilinear involution
\begin{align}
    S_H:H+iH\to H+iH,\qquad h_1+ih_2\mapsto h_1-ih_2,
\end{align}
which is closed because $H$ is closed. $H$ can be recovered from $S_H$ as its fixed point space $H=\ker(S_H-1)$. The polar decomposition of $S_H$ is denoted $S_H=J_H\Delta_H^{1/2}$, where $J_H$ is an antiunitary involution and $\Delta_H$ a positive operator with $\ker\Delta_H=\{0\}$.

We now list a few facts about standard subspaces that are used in the body of the article.

\begin{theorem}\label{theorem:Std}
    Let $H\in\Std(\Hil)$.
    \begin{enumerate}
        \item\label{item:Tomita} {\bfseries\em (Tomita-Takesaki Theorem for standard subspaces)}
        \begin{align}
        J_HH=H',\qquad \Delta_H^{it}H=H,\quad t\in\Rl.
        \end{align}

        \item\label{item:AUH} {\bfseries\em ((Anti-)unitary images)} Let $U\in\AU(\Hil)$ be a unitary or antiunitary operator. Then $UH\in\Std(\Hil)$, and
        \begin{align}
            J_{UH}=UJ_HU^*,\qquad
            \Delta_{UH}=U\Delta_HU^*.
        \end{align}

        \item\label{item:Borchers} {\bfseries\em (Borchers' Theorem)} Let $(V(s))_{s\in\Rl}$ be a unitary one-parameter group with generator $P$ such that $\pm P\geq0$ and $V(s)H\subset H$, $s\geq0$. Then
        \begin{align}
            J_HV(s)J_H=V(-s),\qquad \Delta_H^{it}V(s)\Delta_H^{-it}=V(e^{\mp 2\pi t}s),\quad t,s\in\Rl.
        \end{align}
    \end{enumerate}
\end{theorem}

\section{A Gleason type theorem for symplectic spaces}\label{section:SymplecticGleason}

In this section, we consider a real vector space $V$ with symplectic form $\sigma$ (i.e. a bilinear antisymmetric non-degenerate map $\sigma:V\times V\to\Rl$), and its lattice~$\CL(V)$ of subspaces with symplectic complementation $H\mapsto H'$. As a special case, this also covers the case where $V=\Hil$ is a complex Hilbert space, considered as a real Hilbert space, and symplectic form $\sigma=\imag\langle\,\cdot\,,\,\cdot\,\rangle$. However, none of the additional structure present in this case (complex structure, scalar product, norm) will be needed here. Note that in this purely algebraic setting, $K\vee H=K+H$ (no closure), the symplectic complement needs not be an involution, but $K\subset K''$ holds for any $K\in\CL(V)$, and de Morgan's Laws take the form $(K\vee H)'=K'\cap H'$ and $(K\cap H)'\supset K'\vee H'$.

The aim is to describe the (finitely additive) probability measures on $\CL(V)$ in the sense of Definition~\ref{def:ProbabilityMeasureOnInvolutionLattice}. In preparation, we recall the following facts \cite{deGosson:2006} and introduce some terminology.

\begin{enumerate}[(S1)]
    \item\label{item:Symplectic1} A subspace $H\subset V$ is called {\em symplectic} if $H\cap H'=0$ (this is called {\em factorial} in the standard subspace setting). $H\subset V$ is symplectic if and only if $\sigma|_{H\times H}$ is non-degenerate. Any symplectic space has infinite or finite even dimension. A {\em symplectic plane} is a two-dimensional symplectic subspace $H\in\CL(V)$.

    \item\label{item:Symplectic2} If $\dim H=2n$ for some $n\in\Nl$, then there exist symplectic planes $H_1,\ldots,H_n$ such that
    \begin{align}
        H=H_1\oplus_\sigma\ldots\oplus_\sigma H_n,
    \end{align}
    where $\oplus_\sigma$ denotes an algebraic direct sum that is symplectic, i.e. $H_i\subset H_j'$ for $i\neq j$.

    \item\label{item:Symplectic3} For a general subspace $H\subset V$, the radical of $\sigma|_{H\times H}$ is $Z(H):=H\cap H'$, the center of $H$. There exist (typically non-unique) subspaces $K\subset H$ such that
    \begin{align}
        H
        =
        Z(H)\oplus_\sigma K,
    \end{align}
    such spaces $K$ are necessarily symplectic.

    For $\dim H<\infty$, the {\em symplectic rank} of $H$ is defined as
    \begin{align}
        \rank_\sigma(H)
        :=
        \dim(H)-\dim Z(H)
        =
        \dim(K)
        .
    \end{align}
    As $K$ is symplectic, the symplectic rank of $H$ is always even. Furthermore, $\rank_\sigma$ is additive over symplectic orthogonal direct sums, i.e. $\rank_\sigma(H+K)=\rank_\sigma(H)+\rank_\sigma(K)$ if $H,K$ are finite-dimensional and $H\subset K'$. Furthermore, $K\subset H\Rightarrow\rank_\sigma(K)\leq\rank_\sigma(H)$, i.e. the symplectic rank is monotone.

    In case $V=\Hil$ is a complex Hilbert space with induced symplectic form $\sigma=\imag\langle\,\cdot\,,\,\cdot\,\rangle$, a distinguished subspace of $H$ is the {\em factorial part of} $H$, defined as
    \begin{align}\label{eq:FactorialPart}
        F(H):=Z(H)^\perpr\cap H,
    \end{align}
    and $\rank_\sigma(H)=\dim_{\Rl} F(H)$ holds.
\end{enumerate}
The facts \ref{item:Symplectic1}--\ref{item:Symplectic3} will be used without further mentioning in the following.

\begin{lemma}\label{lemma:Gleason1}
    Let $V$ be a symplectic space, $\mu:\CL(V)\to[0,1]$ a finitely additive probability measure, and $H\in\CL(V)$.
    \begin{enumerate}
        \item\label{item:muOnlySeesFactorialPart} For any symplectic decomposition $H=Z(H)\oplus_\sigma K$, we have $\mu(H)=\mu(K)$.
        \item\label{item:FactorDecomposition} If $H$ is symplectic and decomposed as $H=H_1\oplus_\sigma\ldots\oplus_\sigma H_n$ into symplectic planes $H_k$, then $\mu(H)=\sum_{k=1}^n\mu(H_k)$.
    \end{enumerate}
\end{lemma}
\begin{proof}
    \ref{item:muOnlySeesFactorialPart} The decomposition $H=Z(H)\oplus_\sigma K$ of $H\in\CL(V)$ is symplectic, so $Z(H)\subset K'$. Hence additivity of $\mu$ implies $\mu(H)=\mu(Z(H))+\mu(K)$. Since $Z(H)\subset Z(H)'$, we have $\mu(Z(H))=\mu(Z(H)\vee Z(H))=\mu(Z(H))+\mu(Z(H))$, and therefore $\mu(Z(H))=0$.

    \ref{item:FactorDecomposition} follows immediately from additivity of $\mu$.
\end{proof}

One consequence of this lemma is that any probability measure on $\CL(V)$ is fixed by its restriction to the subset of all symplectic planes in $V$. We next show that this restriction is constant.

\begin{lemma}\label{lemma:Gleason2}
    Let $V$ be a symplectic space, $\mu$ a finitely additive probability measure on~$\CL(V)$, and $K,H\in\CL(V)$ symplectic planes. Then $\mu(K)=\mu(H)$.
\end{lemma}
\begin{proof}
    As $K$ and $H$ are both two-dimensional, we have $\dim(K\cap H)\in\{0,1,2\}$. For $\dim(K\cap H)=2$, we have $K=H$ and the statement is trivial.

    \noindent{\bfseries Case $\mathbf{\dim(K\cap H)=1}$.} In this case, $L:=K\vee H$ has dimension $\dim L=3$ with symplectic rank (even) $\rank_\sigma(L)=3-\dim(L\cap L')\in\{0,2\}$. As $K$ is a symplectic space, $L\neq L'$, so we have $\dim(L\cap L')=1$. Since $\dim K=2$ and
    \begin{align*}
        L\cap L'\cap K=(K+H)\cap K'\cap H'\cap K=0,
    \end{align*}
    we have a decomposition $L=Z(L)\oplus_\sigma K$. By symmetry in $H$ and $K$, we also have $L=Z(L)\oplus_\sigma H$. Now Lemma~\ref{lemma:Gleason1}~\ref{item:muOnlySeesFactorialPart} yields $\mu(K)=\mu(L)=\mu(H)$.

    \noindent{\bfseries Case $\mathbf{\dim(K\cap H)=0}$.} In this case, we will construct symplectic planes $K_1,\ldots,K_n$ (actually $n\in\{1,2\}$ will be sufficient) such that all the intersections $K\cap K_1$, $K_j\cap K_{j+1}$, and $K_n\cap H$ are one-dimensional. Then the proof for the case of a one-dimensional intersection shows that $\mu$ is constant along this chain, and in particular $\mu(K)=\mu(H)$.

    To construct these interpolating planes, we consider two subcases, $K\not\subset H'$ and $K\subset H'$. In case $K\not\subset H'$, there exists $h\in H$ and $k\in K$ such that $\sigma(h,k)\neq0$. We consider $K_1:=\Rl\text{-}\newspan\{h,k\}$. By construction, $K_1$ is a symplectic plane. We have $k\in K_1\cap K$ and $K_1\neq K$ (otherwise $h\in K$, a contradiction to $K\cap H=\{0\}$), so $K$ and $K_1$ are symplectic planes with one-dimensional intersection. In complete analogy, $K_1\cap H$ is one-dimensional, so $\mu(K)=\mu(H)$ follows.

    The other subcase is $K\subset H'$. In this case, the four-dimensional space $L:=K\vee H$ is symplectic: Let $l=k+h\in K+H=L$ (with $k\in K$, $h\in H$) also lie in $L'\subset H'$, i.e. $k+h=h'$ for some $h'\in H'$. As $K\subset H'$ and $H\cap H'=0$, this implies $h=0$. So we have shown $L\cap L'=K\cap L'$, which implies $L\cap L'=K\cap K'\cap H'=0$, as claimed.

    We are thus concerned with a four-dimensional symplectic space and two symplectic planes $K,H\subset L$. These subspaces satisfy $K=H'\cap L$: The inclusion $K\subset H'\cap L$ is clear from $K\subset H'$, and $H'\cap L=H'\cap(K+H)\subset K$ follows from $H\cap H'=0$. Working entirely in $L$, we therefore consider a symplectic plane $K$ and its symplectic complement $H=K'$ (in $L$). Then there exists a symplectic basis $\{e_1,e_2,f_1,f_2\}$ of $L$ (i.e. $\sigma(e_i,f_j)=\delta_{ij}$ and $\sigma(e_i,e_j)=\sigma(f_i,f_j)=0$) such that
    \begin{align}
        K=\newspan\{e_1,f_1\},\qquad H=\newspan\{e_2,f_2\}.
    \end{align}
    We now define $K_1:=\newspan\{e_1+e_2,f_1\}$ and $K_2:=\newspan\{e_1+e_2,f_2\}$. These are clearly both two-dimensional, and $\sigma$ is non-degenerate on them, so $K_1,K_2$ are also symplectic planes. To check the one-dimensional intersections, we note that $f_1\in K\cap K_1$, $e_1+e_2\in K_1\cap K_2$, and $f_2\in K_2\cap H$, so all intersections have dimension at least one. It is easy to check that $K\neq K_1\neq K_2\neq H$. For instance, $K=K_1$ would imply $e_2\in K$, in contradiction to $K\cap H=0$, and similarly for the other cases. This completes the construction of the interpolating chain, and the proof of the lemma.
\end{proof}

\begin{theorem}\label{theorem:SymplecticGleasonTheorem}
    Let $V$ be a real symplectic space.
    \begin{enumerate}
        \item\label{item:MeasureInfiniteDimension} If $V$ is infinite-dimensional, and $\mu$ a finitely additive probability measure on~$\CL(V)$, then $\mu(H)=0$ for all finite-dimensional subspaces $H\subset V$.
        \item\label{item:MeasureFiniteDimension} If $V$ is finite-dimensional, there exists a unique probability measure $\mu$ on $\CL(V)$, namely
        \begin{align}
            \mu(H)
            =
            \frac{\rank_\sigma(H)}{\dim(V)},
            \qquad
            H\in\CL(V).
        \end{align}
    \end{enumerate}
\end{theorem}
\begin{proof}
    Let $\mu_0\in[0,1]$ denote the value that $\mu$ takes on all two-dimensional symplectic subspaces (Lemma~\ref{lemma:Gleason2}), and let $H\in\CL(V)$ be finite-dimensional. Then there exists $n\in\Nl_0$ and symplectic planes $F_1,\ldots,F_n$ such that
    \begin{align*}
        H
        =
        Z(H)\oplus_\sigma F(H)
        =
        Z(H)\oplus_\sigma F_1\oplus_\sigma\ldots\oplus_\sigma F_n.
    \end{align*}
    Lemma~\ref{lemma:Gleason1} yields $\mu(H)=\sum_{k=1}^n\mu(F_k)=n\mu_0=\mu_0\frac{\rank_\sigma(H)}{2}$. In case $V$ is finite-dimensional (with $\dim(V)$ necessarily even), we may apply this to $H=V$ and get $\mu_0=\frac{2}{\dim(V)}$ from $\rank_\sigma(V)=\dim(V)$ and the normalization condition $\mu(V)=1$. This results in the claimed unique map $\mu:H\mapsto\frac{\rank_\sigma(H)}{\dim(V)}$. Clearly $\mu$ maps $\CL(V)$ into $[0,1]$ and is normalized. Its additivity and monotonicity are known properties of the symplectic rank, as recalled before. This proves~\ref{item:MeasureFiniteDimension}.

    In case $V$ is infinite-dimensional, there exist infinitely many pairwise symplectic orthogonal symplectic planes $F_k$, $k\in\Nl$. Then the symplectic spaces $H_n=F_1\oplus_\sigma\ldots\oplus_\sigma F_n$ satisfy $\mu(H_n)=n\mu_0$. As this number must lie in $[0,1]$ for all~$n$, we obtain $\mu_0=0$, i.e. $\mu$ vanishes on all finite-dimensional subspaces of $V$. This proves~\ref{item:MeasureInfiniteDimension}.
\end{proof}

\subsection*{Acknowledgements}

GL would like to thank Carmine De Rosa for discussions about the POVM approach to causal localization observables. Financial support by the Deutsche Forschungsgemeinschaft DFG through the Heisenberg project ``Quantum Fields and Operator Algebras'' (LE 2222/3-1) (GL) and the DAAD project 57588368 for a visit to FAU (IR) is gratefully acknowledged.

\setlength{\biblabelsep}{0.4em}
\AtNextBibliography{\footnotesize}
\newrefcontext[sorting=anyt]
\setlength\bibitemsep{0.25\itemsep}
\printbibliography

\end{document}